\documentclass[10pt, journal]{IEEEtran}

\ifCLASSOPTIONcompsoc
\usepackage[nocompress]{cite}
\else
\usepackage{cite}
\fi
\usepackage{cite,url}
\usepackage{booktabs}
\usepackage{amsmath,amssymb,bm,dsfont}
\usepackage{algorithmic}
\usepackage{amsthm,nccmath,lipsum,cuted}
\usepackage[dvips]{color}
\usepackage{xcolor}
\usepackage{graphicx}
\usepackage[lined,boxed,commentsnumbered, ruled]{algorithm2e}
\usepackage{subcaption}
\usepackage{cleveref}
\usepackage{amsfonts}
\usepackage{ragged2e}
\usepackage{graphicx}
\usepackage{verbatim}
\usepackage[text={6.5in,9in}]{geometry}

\usepackage{epstopdf}

\newtheorem{theorem}{Theorem}
\newtheorem{lemma}{Lemma}

\newtheorem{corollary}{Corollary}

\DeclareMathOperator*{\argmax}{arg\,max}

\usepackage[footnote,draft,danish,silent,nomargin]{fixme}

\def\P{{\mathbb P}}  
\def\E{{\mathbb E}}  

\def\A{A_}
\def\D{D_}

\def\d{{age limit}}
\def\prob{{$\alpha$-UBMP}}
\def\Chernoff{{Chernoff-UBMP}}
\def\nR{{\hat{n}_\text{R}}}

\graphicspath{{./figs/}}

\newcommand{\jpcolor}[1]{{\color{black} #1}}
\newcommand{\jpcol}[1]{{\color{black} #1}}

\begin{document}
	
	\title{Statistical Guarantee Optimization for AoI in Single-Hop and Two-Hop FCFS Systems with Periodic Arrivals}
	\author{Jaya~Prakash~Champati,~\IEEEmembership{Member,~IEEE,}
		Hussein~Al-Zubaidy,~\IEEEmembership{Senior~Member,~IEEE,}
		and~James~Gross,~\IEEEmembership{Senior~Member,~IEEE}
		\thanks{J.P. Champati, H. Al-Zubaidy, and J. Gross are with the School of Electrical Engineering and Computer Science, KTH Royal Institute of Technology, Stockholm,
			Sweden (e-mail: \{jpra,hzubaidy,jamesgr\}@kth.se).
			This research was supported in part by the Swedish Research Council (VR) under grant 2016-04404. This work is recently published in IEEE Transactions on Communications.}
	}
	
	\makeatletter
	\def\ps@IEEEtitlepagestyle{
		\def\@oddfoot{\mycopyrightnotice}
		\def\@evenfoot{}
	}
	\def\mycopyrightnotice{
		{\footnotesize
			\begin{minipage}{\textwidth}
				\centering
				0090-6778~\copyright 2020 IEEE. Personal use is permitted, but republication/redistribution requires IEEE permission.\\
				See https://www.ieee.org/publications/rights/index.html for more information.
			\end{minipage}
		}
	}
	\maketitle
	
	\begin{abstract}
		Age of Information (AoI) has proven to be a useful metric in networked systems where timely information updates are of importance. 
		In the literature, minimizing ``average age" has received considerable attention. 
		However, various applications pose stricter age requirements on the updates which demand knowledge of the AoI distribution.  Furthermore, the analysis of AoI distribution in a multi-hop setting, which is important for the study of Wireless Networked Control Systems (WNCS), has not been addressed before.   
		Toward this end, we study the distribution of AoI in a WNCS with two hops and devise a problem of minimizing the tail of the AoI distribution with respect to the frequency of generating information updates, i.e., the sampling rate of monitoring a process, under first-come-first-serve (FCFS) queuing discipline. 
		We argue that computing an exact expression for the AoI distribution may not always be feasible; therefore, we opt for computing upper bounds on the tail of the AoI distribution. Using these upper bounds, we formulate Upper Bound Minimization Problems (UBMP), namely, Chernoff-UBMP and $\alpha$-relaxed Upper Bound Minimization Problem (\prob), where $\alpha > 1$ is an approximation factor, and solve them to obtain ``good'' heuristic rate solutions for minimizing the tail. 
		We demonstrate the efficacy of our approach by solving the proposed UBMPs for three service distributions: geometric, exponential, and Erlang. 
		Simulation results show that the rate solutions obtained are near optimal for minimizing the tail of the AoI distribution for the considered distributions.    
	\end{abstract}
	

	\section{Introduction}
	In the recent past, there has been an ever increasing interest in studying Wireless Networked Control Systems (WNCS) that support time-critical-control applications which include, among many others, autonomous vehicle systems, automation of manufacturing processes, smart grids, Internet-of-Things (IoT), sensor networks and augmented reality.
	A basic building block in WNCS is depicted in Figure~\ref{fig:NCS}. A sensor samples a plant/process of interest and transmits the status updates or packets over a wireless channel (link $1$) to a controller. The controller computes a control input using the received status update and transmits it to an actuator, using another communication channel (link $2$). 
	A status update that is received at the controller after a certain duration of its generation time may become stale, and the control decision taken based on this stale sample may result in untimely actuation affecting the performance of a time-critical-control application in a WNCS. Similarly, the same effect could result from a control decision (based on a fresh status update) reaching the actuator after a delay deadline. 
	In this respect, the traditional goal of maximizing throughput becomes less relevant as freshness of the status updates not only depends on queuing and transmission delays in the network, but also on the frequency of generating updates at the source. 
	
	\begin{figure}
		\centering
		\includegraphics[width = 3.2in]{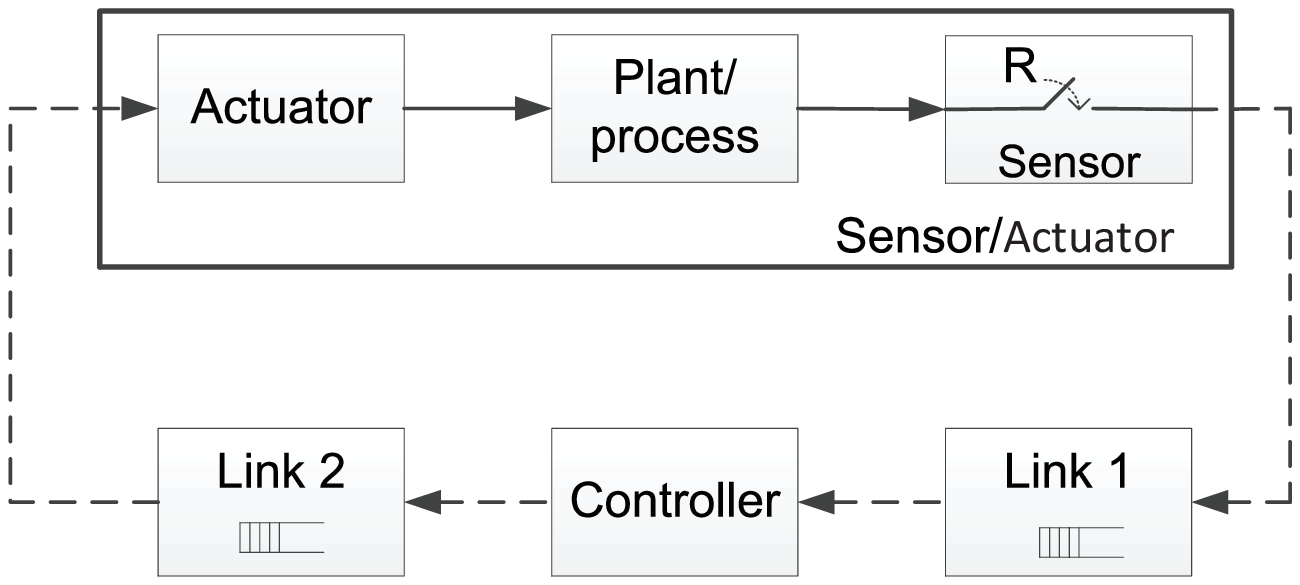}
		\caption{A networked control system with a remote controller.}
		\label{fig:NCS}
		\vspace{-.5cm}
	\end{figure}
	
	Age of Information (AoI), proposed in~\cite{kaul_2011a}, has emerged as a relevant performance metric in quantifying the freshness of the status updates at a destination. It is defined as the time elapsed since the generation of the latest status update received at the destination. AoI accounts for the frequency of generation of updates by the source, since it linearly increases with time until a status update with latest generation time is received at the destination. Whenever such an update is received, AoI resets to the system delay of the update indicating its age. 
	\jpcolor{Motivated by the fact that having access to fresher status updates improves the control performance in WNCS, we model the control network by a two-hop FCFS queuing system  and formulate a problem of computing optimal sampling rate that minimizes AoI in this system\footnote{A preliminary version of this work only considering single-hop scenario appeared in~\cite{Champati_DG1_2018}.}. Several research works in the recent past addressed the problem of optimizing sampling rate in different queuing systems under various settings. However, as we explain in Section~\ref{sec:related}, these works either consider a single-hop system or memoryless arrivals or some form of ``average age'' function. In contrast, we consider two novel aspects that are relevant to time-critical-control applications. First, we consider periodic arrival process by assuming that the process of interest is sampled at a constant rate $R$. This is motivated by the fact that sensors in practice are typically configured to generate samples periodically. Second, since optimizing average statistics of AoI may not meet stringent QoS requirements, for instance, in a safety-critical system~\cite{SafetyRequirements}, we consider optimizing \textit{AoI violation probability}, i.e., the probability that AoI at the actuator violates a given \textit{age limit} $d$. This metric represents, for example, a reliability measure constraint that is required at the actuator to insure that the state of the plant  remains within a predetermined safety boundary. Furthermore, in a WNCS, an absolute guarantee (i.e., reliability of 1) may not be possible due to variability of the wireless channel and only probabilistic guarantees can be provided. This motivated us to use the distribution of AoI as a metric rather than other frequently used metrics in the literature, e.g., peak AoI and average AoI. 
		
		We consider a heterogeneous network, i.e., server at the first queue and server at the second queue may have different service-time distributions. The queues operate using First-Come-First-Serve (FCFS) scheduling discipline. \jpcolor{We note that, in the AoI literature, different scheduling disciplines are considered: for example~\cite{kaul_2012b,Huang_2015,Talak_2018b} considered FCFS,~\cite{kaul_2012a,Bedewy_2017a,Yates2018} considered LCFS, and~\cite{Costa_2016,Champati_GG1_2019} considered packet management schemes such as using a unit capacity queue with packet replacement. 
			{\color{black}
				Our motivation for considering FCFS discipline in this work is the following. First, analysis and optimization of AoI violation probability under FCFS is an open problem. Second, it is not only an interesting problem from academic (queuing-theoretic) point of view, due to FCFS being more intuitive and hence such analysis being more comprehensible, but also important from practical point of view as most queues in practice operate under FCFS.
				Third,  key insights, e.g.,  as the sampling rate $R$ increases, in contrast to \textit{delay}, AoI first decreases and then increases~\cite{kaul_2012b}, that are established under FCFS discipline may be extended to other disciplines as well.
				Lastly, the analysis of many important queuing disciplines can be based on, and sometimes directly derived from, that of FCFS, i.e., by introducing a queue reordering stage based on arrival instance, priority, or some fairness parameter before serving the head of the queue. We believe that our analysis can potentially be extended to such queuing disciplines in future works.   
			}
		}  
		
		

		Assuming that the processing time at the controller is negligible, we aim to compute $R$ that minimizes the AoI \textit{violation probability} at the egress of the second queue. As we will see in a while, an exact expression for the AoI violation probability in a two-hop network with periodic arrivals and general service-time distributions is intractable.} Therefore, we resort to working with tractable upper bounds which facilitate the computation of ``good'' heuristic solutions. In particular, we first compute the upper bounds for the single-hop case, i.e., the D/G/1 queue, due to its relevance in applications where both controller and actuator are collocated. We formulate the Upper Bound Minimization Problems (UBMP) and \jpcolor{use them to compute heuristic rate solutions for AoI violation probability minimization}. We then extend the results for two-hop and N-hop tandem queuing systems using max-plus convolution for the service processes.


	The main contributions of this work are summarized below:
	\begin{itemize}
		\item \jpcolor{We characterize the probability that AoI violates a given \d~$d$ for a single source single destination multi-hop network under FCFS, and assuming periodic source where packets are generated at a constant rate $R$.}
		\item We formulate the AoI violation probability minimization problem $\mathcal{P}$, and show that it is equivalent to minimizing the violation probability of the departure instant of a \jpcolor{tagged packet (defined in Section V)} over the rate region $[\frac{1}{d},\mu)$, where $\mu$ is the service capacity of the network. 
		\item Using the above characterization, we first propose a UBMP for the single-hop scenario, i.e., the D/G/1 queue. Noting that the objective function in the UBMP can be intractable, we propose a Chernoff-UBMP, that has a closed-form objective, and an $\alpha$-relaxed UBMP the solution of which has $\alpha > 1$ approximation ratio \jpcolor{(worst-case ratio)} with respect to \jpcolor{the objective function of the} UBMP.
		\item We extend the derived results and formulations for the two-hop queuing system and $N$-hop tandem queuing system, and present example computation of the expressions for the case of two-hop for geometric, exponential, and Erlang service-time distributions.
		\item We demonstrate the efficacy of the heuristic solutions provided by Chernoff-UBMP and $\alpha$-relaxed UBMP using simulation for different service-time distributions. \jpcol{Finally, we present simulation results comparing the performance of FCFS with queue management policies that use unit buffer and packet replacement.}
	\end{itemize}
	
	The rest of the paper is organized as follows. In Section~\ref{sec:related}, we present the related work. In Section~\ref{sec:model}, we present the problem formulation. Analysis of the AoI violation probability is presented in Section~\ref{sec:vioprob}. The UBMP formulations for single-hop, two-hop and N-hop scenarios are presented in Sections~\ref{sec:singleHop} and~\ref{sec:twohop}, respectively. We present the computation of the upper bounds for different service-time distributions in Section~\ref{sec:exampleDis}. Numerical results are presented in Section~\ref{sec:numerical} and we finally conclude in Section~\ref{sec:conclusion}.

	\section{Related Work}\label{sec:related}
	Several works in the AoI literature have focused on analyzing and providing expressions for average AoI statistics in different queuing systems, e.g., see~\cite{kaul_2012a,Chen2016,Najm_2016,Najm_2017,Soysal2019}. \jpcolor{The authors in~\cite{Costa_2016} studied the M/M/1/1 and M/M/1/2*\footnote{A unit capacity queue that holds the latest update.} systems, and computed the average AoI and the distribution of the peak AoI. In contrast, the authors in~\cite{Yoshiaki2018,Champati_GG1_2019} provided expressions for the distribution of AoI. However, for the case of periodic arrivals, closed-form expressions are provided only for single-hop scenario and for exponential service times in~\cite{Yoshiaki2018}, and for the case of no queue in~\cite{Champati_GG1_2019}. Next, we summarize works that consider optimizing AoI under different system settings. An interested reader may also refer to~\cite{KostaMonograph_2017} and \cite{SunMonograph2020} for a comprehensive survey of recent work in this area.
		
		
		In~\cite{kaul_2012b}, the authors have addressed the problem of computing the optimal arrival rate to minimize the \textit{time-average age} for M/M/1, M/D/1 and D/M/1 queuing systems. This problem was addressed for M/M/1 with multiple sources in~\cite{yates_2012a}. Several research works that followed considered different design choices including the arrival rate~\cite{Huang_2015}, inter-arrival time distribution for a given arrival rate and/or service-time distribution for a given service rate~\cite{Talak_2018a,Soysal2019,Talak2019a,Talak2019b}, under different scheduling disciplines and optimized average AoI or average peak AoI in a single-source-single-server system. \jpcol{In~\cite{Bedewy2019}, preemptive Last Generated First Served (LGFS) policy was shown to minimize the age process in a multi-server single-hop system with exponential service times.}  
		An alternative approach to the above works, the generate-at-will source model was studied in~\cite{yates_2015a,Sun_2017,Champati_2020}, where generation of a status update can be completely controlled. While the authors in~\cite{yates_2015a} solved for optimal-waiting times between generation times to minimize the average AoI, the authors in~\cite{Sun_2017} solved the problem for any non-decreasing function of AoI, and the authors in~\cite{Champati_2020} solved the problem of minimum achievable peak AoI in any single-source-single-server system. The authors in~\cite{Talak_2018b} studied average AoI and average peak AoI minimization for multiple source-destination links in a wireless network with interference constraints. They used the method of minimizing upper bounds as a means to show that optimal rate design and optimal link scheduling can be separated and provided performance guarantees for the proposed solutions.
		
		In addition to the above, the following literature considered multi-hop settings. \jpcol{For a line network with a single source and no queues, under Poisson arrivals and exponential service times, expressions were derived in~\cite{Yates2018} for moments, Moment Generating Function (MGF), and stationary distribution of AoI for preemptive last-come-first-served policy. In~\cite{Bedewy_2017a}, optimal queuing policies were investigated for a multi-hop network for any arrival sequence and service-time distributions. It was shown that, among non-preemptive policies, LGFS minimizes age processes, in stochastic ordering sense, at all the nodes.} The authors in~\cite{Talak2017,Farazi2019} studied average AoI and average peak AoI minimization in a multi-hop wireless network with interference constraints and with packet flows between multiple source-destination pairs assuming that transmission time of a packet equals a unit time slot. The authors in~\cite{Buyukates2018} studied average AoI for $L$-hop multicast network with a single source, $n$ nodes in the first hop, $n^2$ nodes in the second hop, and so on $n^L$ nodes in the the last hop, with each node having a shifted exponential service time.}
	
	Optimizing AoI was also extensively studied for the systems with energy-harvesting source, e.g., see~\cite{yates_2015a,Bacinoglu_2015a,Bacinoglu_2019}. In the context of a cloud gaming system the authors in~\cite{Yates2017} used the D/G/1 system model to study the effect of freshness on video frame rendering to the client. Specifically, they have analyzed the average age by considering the aspect of missing frames. In contrast to all the above works, with motivations from the sensor-controller-actuator system in WNCS we study the problem of AoI violation probability minimization in a two-hop queuing system with periodic arrivals.

	\section{System Model and Problem Statement}\label{sec:model}
	\label{model}
	Motivated by the sensor-controller-actuator communicating over wireless channels, we study a two-hop queuing system, shown in Figure~\ref{fig:model}, under FCFS scheduling. The source generates packets (status updates) at a constant rate $R$. 
	Thus, $R$ models the sampling rate of a process under observation. Let $T = \frac{1}{R}$ denote the inter-arrival time between any two packets. We index the nodes by $k \in \{1,2\}$, and the packets by $n \in \{0,1,2\ldots\}$. 
	Let $\A{{k}}(n,R)$ denote the arrival instant of packet $n$ and $\D{{k}}(n,R)$ the corresponding departure instant at node $k$. For notational simplicity, we use $A(n,R) = A_1(n,R)$ and $D(n,R) = \D{{2}}(n,R)$ to denote the arrivals and departures of the system, respectively. Also, we have $\A{{2}}(n,R) = \D{{1}}(n,R)$.
	The arrival time of packet $n$ to the system is given by $A(n,R) = \frac{n}{R}$. 
	The service time for packet $n$ at node $k$ is given by a random variable $X_{k}^n$. For $k \in \{1,2\}$, we assume $X_{k}^n$ are i.i.d., for all $n$, with mean service rate $\mu_{k} = \frac{1}{\E[X_{k}^1]} > 0$. 
	Also, we assume that $X_{1}^n$ and $X_{2}^n$ are independent, for all $n$, but may have non-identical distributions, i.e., the servers could be heterogeneous. We define $\mu \triangleq \min(\mu_{1},\mu_{2})$. Later, in Section~\ref{sec:multihop}, we show how the results can be extended to $N$-hop tandem queuing network.
	
	\begin{figure}
		\centering
		\includegraphics[width = 3.2in]{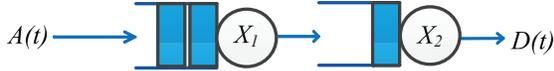}
		\caption{Model of the two-hop network.}
		\label{fig:model}
	\end{figure}
	
	At the destination, we are interested in maintaining timely state information of the process.
	We are thus interested in the AoI metric, denoted by $\Delta(t,R)$, which is defined as:
	\begin{align}\label{eq:AoI-Definition}
	\Delta(t,R) \triangleq t - \max_n\{A(n,R): D(n,R) \leq t\}. 
	\end{align}
	
	For a given \textit{age limit} requirement $d > 0$, in the following we study the distribution of AoI by characterizing its violation probability, i.e., $\P(\Delta(t,R) > d)$, both in the transient and the steady states of the system. 
	Given the age limit $d$, we are interested in solving the following problem $\mathcal{P}$:
	\begin{equation*}
	\begin{aligned}
	& \underset{R}{\text{min}}\: \lim_{t\rightarrow \infty} \P(\Delta(t,R) > d).
	\end{aligned}
	\end{equation*}
	Let $R^*(d)$ denote an optimal rate solution for $\mathcal{P}$. \jpcolor{In the sequel, we refer to $\lim_{t\rightarrow \infty} \P(\Delta(t,R) > d)$ as \textit{AoI violation probability}.}

	Henceforth, we drop $R$ from the notation when it is obvious from the context, for the sake of notation simplicity. For $k \in \{1,2\}$, the MGF of $X^n_{k}$ is given by $M_{k}(s) = \mathbb{E}[e^{sX^n_{k}}]$.

	We now state the Chernoff bound, which will be used extensively to formulate the upper bound minimization problems in the sequel. \jpcolor{Assuming that the moment generating function of a random variable $Y$ exists, the \textbf{Chernoff bound} for its distribution is given by
		\begin{align*}
		\P\{Y > y\} \leq \min_{s > 0}\, e^{-sy} \mathbb{E}[e^{sY}].
		\end{align*}}
	Note that the upper bounds derived using the Chernoff bound involves minimization over the parameter $s$. We shall see that, for the two-hop network, these bounds attain finite values only when there exists $s > 0$ such that $\max(M_{1}(s),M_{2}(s)) < e^{s/R}$. To this end, we formulate the minimization problems over the set $\mathcal{S} \subseteq \mathbb{R}^+$ which characterizes $s$ values for which $\max(M_{1}(s),M_{2}(s)) < e^{s/R}$, i.e.,
	\begin{align}\label{eq:calS}
	\mathcal{S} \triangleq \{s > 0: \max(M_{1}(s),M_{2}(s)) < e^{s/R}\}.
	\end{align}
	We assume that $\mathcal{S}$ is non-empty. In the following lemma we show that this assumption is in fact  a sufficient condition for the stability of the system.
	\begin{lemma}\label{lem:queueStability}
		If there exists $s > 0$ such that 
		\begin{align*}
		\max(M_{1}(s),M_{2}(s)) < e^{s/R},
		\end{align*}
		then the queues are stable.
	\end{lemma}
	\begin{proof}
		Recall that the queues are stable if $\min(\mu_{1},\mu_{2}) > R $. Consider the case $M_{1}(s) < e^{s/R}$, which implies
		\begin{align*}
		\mathbb{E}[e^{sX^n_{1}}] < e^{s/R} 
		\Rightarrow e^{s\mathbb{E}[X^1_{1}]} < e^{s/R}
		\Rightarrow \mu_{1} > R,
		\end{align*}
		for any $s > 0$. In the second step above we have used Jensen's inequality. Similarly, if $M_{2}(s) < e^{s/R}$, then $\mu_{2} > R$. Therefore, \jpcolor{if there exists $s > 0$ such that} $\max(M_{1}(s),M_{2}(s)) < e^{s/R}$, then  $\min(\mu_{1},\mu_{2}) > R $, and the lemma follows.
	\end{proof}
	We define 
	\begin{align}\label{eq:beta}
	\beta_{k}(s) \triangleq \frac{M_{k}(s)}{e^{s/R}},\, k \in \{1,2\}.
	\end{align}
	By definition, for all $s \in \mathcal{S}$, $\beta_{k}(s) < 1$. \jpcolor{The list of symbols used in the paper are summarized in Table~\ref{tabel1}. }
	
	\begin{table}[ht]
		\renewcommand{\arraystretch}{1.2}
		\caption{\jpcolor{List of Symbols}}
		\centering
		\begin{tabular}{|l|c|c|}
			\hline
			$k$ & Node/link index \\
			\hline
			$N$ & Number of nodes \\
			\hline
			$n$ & Packet index \\
			\hline
			$R$ & Sampling rate \\
			\hline
			$T$ & Inter-arrival time ($\frac{1}{R}$) \\
			\hline
			$A_k(n,R)$ & Arrival time of packet $n$ at node $k$\\
			\hline
			$D_k(n,R)$ & Departure time of packet $n$ at node $k$\\
			\hline
			$A(n,R)$ & Arrival time of packet $n$ in the system ($\frac{n}{R}$)\\
			\hline
			$D(n,R)$ & Departure time of packet $n$ from the system\\
			\hline
			$X^n_k$  & Service time of packet $n$ at node $k$\\
			\hline
			$\mu_k$  & Service rate at node $k$\\
			\hline
			$\Delta(t,R)$  & Age of information at time $t$\\
			\hline
			$d$  & Age limit\\
			\hline
			$M_k(\cdot)$  & MGF of service at node $k$\\
			\hline
			$\hat{n}_R$  & Index of the first arrival after time $t-d$\\
			\hline
		\end{tabular}
		\label{tabel1}
	\end{table}

	\section{AoI Violation Probability Analysis}\label{sec:vioprob}
	\jpcolor{In this section, we study the properties of the distribution of AoI -- the results derived are valid for any number of nodes in tandem between the source and the destination given that the packets are input to the network by the source at a constant rate $R$ and the network uses FCFS.} 
	
	We start by investigating  structural  characteristics of the stochastic behaviour of AoI. Toward this end, we use the max-plus representation of Reich's equation to model the evolution of the queues.
	For any realization of the service times at node $k$, the relation between $\D{{k}}(n,R)$, $\A{{k}}(n,R)$ and $\{X_{{k}}^n\}$,  is given by~\cite{JorgBook2017}:
	\begin{align}\label{eq:departureTime}
	\D{{k}}(n,R) = \max_{0 \leq v \leq n} \{\A{{k}}(n-v,R) + \sum_{i=0}^{v}X_{{k}}^{n-i}\}.
	\end{align}
	\jpcolor{We note that equation~\eqref{eq:departureTime} is a direct consequence of using recursion on a fundamental relation in queuing system: $\D{{k}}(n,R) = \max\{\D{{k}}(n-1,R), A_k(n)\} + X^n_k$, which states that the departure time of packet $n$ is given by either its service time plus departure time of previous packet $n-1$ or arrival time of packet $n$ plus its service time, whichever is greater.}
	
	Consider the definition in \eqref{eq:AoI-Definition}, for $\Delta(t,R)$ not to exceed the age limit $d$, the latest departure at $t$ must have arrived no earlier than $t-d$. Therefore, to study the distribution of  $\Delta(t,R)$, we tag the packet arriving on or immediately after $t-d$ and use it to characterize this process. 
	Given rate $R$, let $\nR$ denote the \jpcolor{index of the first arrival since time $t-d$}, given by
	\begin{align}\label{eq:nR}
	\nR \triangleq \lceil R(t-d) \rceil.
	\end{align}
	The tagged packet\footnote{$\nR$ is a function of $t-d$ as well. We omit $t-d$ from the notation here for ease of exposition. }  $\nR$  plays a key role in characterizing the violation probability as we will show next. 
	
	In the following lemma we present a key insight regarding the transient characterization of the AoI violation probability.
	\begin{lemma}
		\label{lem1}
		Given the input arrival rate $R$, age limit $d$, and $t < \infty$, if there exists $n$ such that $t-d \leq \frac{n}{R} < t$, then $\P\{\Delta(t,R) > d\} = \P\{D(\nR) > t\}$, otherwise, $\P\{\Delta(t,R) > d\} = 1$.
	\end{lemma}
	\begin{proof}
		Let $n^*_\text{R}$ be the latest packet departure at $t$, i.e., $n^*_\text{R} = \argmax_{n} \{D(n,R) \leq t\}$. Thus, 
		$\Delta(t,R) = t - A(n^*_\text{R}).$
		
		\textbf{Case 1:} 
		If  an $n$ such that $t-d \leq \frac{n}{R} < t$  does not exist, i.e., there is no arrival during the time interval $[t-d,t)$, then  the arrival time of $n^*_\text{R}$ must be strictly less than $t-d$, i.e., $A(n^*_\text{R}) < t-d$. Therefore,
		\begin{align*}
		\P(\Delta(t,R) > d) = \P(t - A(n^*_\text{R}) > d) =  1.
		\end{align*}
		
		\begin{figure}
			\centering
			\includegraphics[width = 2.8in]{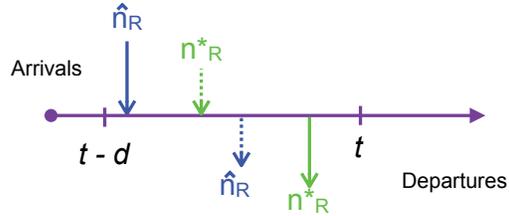}
			\caption{Time-line of events for Case 2 in Lemma \ref{lem1} proof. }
			\label{fig:Case2}
		\end{figure}
		
		\textbf{Case 2:} If there exists $n$ such that $t-d \leq \frac{n}{R} < t$, then $t-d \leq \frac{\nR}{R} < t$, since $\nR$ is the first arrival on or after time $t-d$, see Figure~\ref{fig:Case2}. In this case, we show that the event $\{\Delta(t,R) \leq d\}$ is equivalent to the event $\{D(\nR) \leq t\}$. Suppose that the event $\{\Delta(t,R) \leq d\}$ occurred, then $A(n^*_\text{R}) \geq t - d$. By definition of $\nR$, we should have $A(\nR) \leq A(n^*_\text{R})$ which implies $D(\nR) \leq D(n^*_\text{R}) \leq t$, due to FCFS assumption. 
		Therefore, 
		\begin{equation}\label{eq:Lemma1-case2}
		\{\Delta(t,R) \leq d\} \subseteq \{D(\nR) \leq t\}.
		\end{equation} 
		
		To prove equivalence of the two events, we show that the relation above also holds the other way around.
		Suppose that the event $\{D(\nR) \leq t\}$ occurred. Again, it should be true that $A(n^*_\text{R}) \geq A(\nR)$. Otherwise, $D(n^*_\text{R}) < D(\nR) \leq t$ which contradicts the definition of $n^*_\text{R}$ that it is the latest departure before $t$. Therefore,
		\begin{align*}
		\Delta(t,R) = t - A(n^*_\text{R}) \leq t - A(\nR) \leq t - (t-d) = d.  
		\end{align*} 
		This implies  that $\{D(\nR) \leq t\} \subseteq \{\Delta(t,R) \leq d\}$. Therefore, the equivalence holds and the result is proven.
	\end{proof}

	\jpcolor{The intuition behind the result in Lemma~\ref{lem1} is that AoI exceeds $d$ at time $t$ if none of the packets generated/arrived in the time interval $[t-d,t]$ have reached the destination at time $t$ or no packets are generated in this interval. Note that \textbf{Case 1} in the above proof essentially represents an under-sampling of the process under observation, i.e., at the current time $t$ the sampling rate $R$ is simply too low such that there is no packet generated in the time interval $[t-d,t]$.}  
	
	We next present the steady-state results for the two-hop system based on the result obtained in Lemma \ref{lem1}.
	\begin{theorem}\label{thm:steadystate}
		Given age limit $d$, the steady state distribution of AoI is characterized as follows: 
		\begin{enumerate}
			\item If $R \geq \frac{1}{d}$, then
			\begin{align}\label{eq:steadystateUB}
			\!\!\!\lim_{t \rightarrow \infty} \P \{\Delta(t,R) > d\} = \lim_{t \rightarrow \infty} \P\{D(\nR) > t\}.
			\end{align}
			\item Else if $R < \frac{1}{d}$, then
			\begin{align*}
			& \limsup_{t \rightarrow \infty} \P\{\Delta(t,R) > d\} = 1, \\
			& \liminf_{t \rightarrow \infty} \P\{\Delta(t,R) > d\} = \lim_{t \rightarrow \infty} \P\{D(\nR) > t\}.
			\end{align*}
		\end{enumerate}
	\end{theorem}
	
	\begin{proof}
		For the two cases above consider the following:
		
		\textbf{Case 1 ($\mathbf{R \geq \frac{1}{d}}$):} Since the samples are generated at a constant rate, for $R \geq \frac{1}{d}$ we claim that there exist an $n$ such that $t-d \leq \frac{n}{R} < t$, for all $t$. We first prove this claim for $R > \frac{1}{d}$. We have
		\begin{align*}
		A(\nR) = \frac{\lceil R(t-d) \rceil}{R} \leq  \frac{R(t-d) + 1}{R} < t\, .
		\end{align*}
		Furthermore, since $t-d \leq A(\nR)$ for any $t$ by definition, the claim holds at least for $\nR$, for $R > \frac{1}{d}$. To prove the claim for $R = \frac{1}{d}$, we consider
		\begin{align*}
		t-d \leq\! \frac{n}{R} \! <\! t \, \Leftrightarrow \,  \frac{n}{R}\! < \! t  \leq d + \frac{n}{R} \, \Leftrightarrow \,  n \! < Rt \leq \! n + 1.
		\end{align*}
		Note that for any $R$ and $t$ there always exists an $n$ such that the last inequality above holds. Therefore, the claim is true and Case~1 follows  from Lemma~\ref{lem1} by letting $t$ go to infinity.
		
		\textbf{Case 2 ($\mathbf{R < \frac{1}{d}}$):} In this case, the existence of $n$ such that $t-d \leq \frac{n}{R} < t$ depends on $t$. \jpcolor{To see this, for a given $n$ consider the two time intervals $(\frac{n}{R},\frac{n}{R}+d]$ and $(\frac{n}{R}+d,\frac{n+1}{R})$. Note that the latter time interval is non-empty because $d < \frac{1}{R}$. Now, for time instants $t \in (\frac{n}{R},\frac{n}{R}+d]$ we have $t-d \leq \frac{n}{R} < t$, and therefore for such $t$ using Lemma~\ref{lem1} we have $\P\{\Delta(t,R) > d\} = \P\{D(\nR) > t\}$. On the other hand, for time instants $t \in (\frac{n}{R}+d,\frac{n+1}{R})$, there is no $n$ value such that $t-d \leq \frac{n}{R} < t$ is true, and therefore from Lemma~\ref{lem1} we have $\P\{\Delta(t,R) > d\} = 1$. This implies that as $t$ goes to infinity the violation probability either equals $\P\{D(\nR) > t\}$ or $1$ depending on the value of $t$.} Thus, we obtain the limit supremum and the limit infimum. 
	\end{proof}
	
	{\color{black} Intuitively, given $R$, the support of the steady state AoI distribution should be $[\frac{1}{R},\infty)$, because AoI cannot be less than $\frac{1}{R}$ when the samples are generated at rate $R$. Not only Theorem~\ref{thm:steadystate} asserts this intuitive reasoning, but also characterizes the limit infimum and limit supremum for the region $d < \frac{1}{R}$, where the AoI violation probability does not exist.}
	Therefore, to ensure the existence of the AoI violation probability 
	we consider the feasible rate region $[\frac{1}{d},\mu)$, where 
	$\mu = \min(\mu_1,\mu_2)$, 
	and $R < \mu$ ensures queue stability.
	In light of this, and using~\eqref{eq:steadystateUB} from Theorem~\ref{thm:steadystate}, we formulate an equivalent problem 
	$\mathcal{\tilde P}$ as follows:
	\begin{equation}\label{equivalentProb}
	\begin{aligned}
	& \underset{\frac{1}{d} \leq R < \mu}
	{\text{min}}\: \lim_{t\rightarrow \infty} \P(D(\nR) > t).
	\end{aligned}
	\end{equation}
	
	\textbf{\textit{Remark 1:}} \jpcolor{The results in Lemma~\ref{lem1} and Theorem~\ref{thm:steadystate} are valid for arbitrary single-source single-destination network with constant  arrival rate $R$ and using FCFS queuing discipline, i.e.,
		packets should be received at the destination in the same order as they are transmitted by the source.} For arbitrary network topology, one can formulate problem $\tilde{\mathcal{P}}$ given in~\eqref{equivalentProb} with the following constraints on R: 1) $R \geq \frac{1}{d}$, and 2) $R$ belongs to the rate region in which the network is stable.
	
	Next, we present our solution approach for solving $\tilde{\mathcal{P}}$ for a single-hop case and then show how the approach can be extended for the two-hop system in Section~\ref{sec:twohop}.

	\section{Single-Hop Scenario}\label{sec:singleHop}
	In this section we solve $\mathcal {\tilde P}$ by assuming that $X^n_{2} = 0$ for all $n$. This implies that $D(n) = \D{{1}}(n)$, $\mu_{2} = \infty$, and the system is equivalent to the D/GI/1 system. Our motivation for presenting the single-hop case is because of its importance in solving the two-hop case, and also due to its relevance to practical scenarios, where only estimation of the processes is required, or both controller and actuator are collocated.
	In order to find a solution for $\mathcal {\tilde P}$, we  must first evaluate the probability $\P\{D(\nR) > t\}$, where $D(n)$ is given by \eqref{eq:departureTime}.
	Note that $D(n)$ is random, since the service process $\{X_{1}^n,n \geq 0\}$ is random, and is given in terms of   the maximum of $n+1$ random variables. 
	Hence, obtaining an exact expression is tedious. 
	Therefore, we opt for a more tractable approach by using probabilistic inequalities to obtain bounds on the distribution of  $D(\nR)$. 
	Consequently, we propose the Upper Bound Minimization Problem (UBMP) and its more computationally tractable counterparts \prob~and Chernoff-UBMP to obtain near optimal heuristic solutions for $\mathcal{\tilde P}$.
	\subsection{A Bound for the Distribution of $D$}
	As mentioned earlier, the evaluation of the distribution function of $D(n)$ requires the computation of the distribution of the maximum of random variables. Fortunately, there are several approaches that have been used in the literature to estimate this probability. One such approach approximates the probability of the maximum by the maximum probability, i.e., $\P\{\max_i Y_i > y\} \approx \max \, \P\{ Y_i>y \}$. However,  this approximation is not always accurate and in some cases may result in very large deviation from the actual distribution. Hence, it cannot be used when reliability of the solution must be well defined as it is the case here. An alternative approach is to use extreme value theorem. However, the obtained extreme value distributions are not always tractable. A more promising approach is to use Boole's inequality, commonly known as the ``union bound,'' where the probability of a union of events is bounded by the sum of their probabilities. The bound obtained in our case is not only tractable, but also provides good heuristic solutions for $\mathcal {\tilde P}$. 
	In the following lemma, we present this upper bound for the distribution function $\lim_{t \rightarrow \infty} \P\{D(\nR) > t\}$.
	
	\begin{lemma}\label{lem:singleHopUB}
		Given $d$, we have
		\vspace{-.3cm}
		\begin{align*}
		\lim_{t\rightarrow \infty} \P(D(\nR) > t) \leq \sum_{v=0}^{\infty} \Phi(v,R),
		\vspace{-.5cm}
		\end{align*}
		where
		\vspace{-.3cm}
		\begin{align}\label{eq:Phi:singlehop}
		\Phi(v,R) \triangleq \P\left\{\sum_{i=0}^{v} X_{1}^i > d + \frac{v-1}{R}\right\}.
		\end{align}
	\end{lemma}
	\begin{proof}
		Using~\eqref{eq:departureTime}, we have 
		\allowdisplaybreaks {\begin{align*}
			&\P\{D(\nR)\hspace{-.1cm} > \hspace{-.1cm}t\}\! = \!\P \hspace{-.1cm}\left\{\!\max_{0 \leq v \leq \nR} \hspace{-.1cm} \left (\!\!A(\nR-v)\hspace{-.1cm} + \hspace{-.1cm}\sum_{i=0}^{v}X_{1}^{\nR-i}\right) \hspace{-.1cm} > \hspace{-.1cm} t \right\} \\
			&\overset{(a)}{\leq} \sum_{v=0}^{\nR} \P\left\{\sum_{i=0}^{v} X_{1}^{\nR-i} > t - \frac{\nR-v}{R}\right\}\\
			& \overset{(b)}{\leq} \sum_{v=0}^{\nR} \P\left\{\sum_{i=0}^{v} X_{1}^{\nR-i} > t - \frac{R(t-d) + 1-v}{R}\right\} \\
			&= \sum_{v=0}^{\nR} \,\,\underbrace{  \P\left\{\sum_{i=0}^{v} X_{1}^i > d + \frac{v-1}{R}\right\} }_{ \triangleq \, \Phi(v,R)}.
			\end{align*}}
		\noindent In step (a) we have applied the union bound, and used $\nR = \lceil R(t - d) \rceil \leq R(t-d) + 1$ in step (b). The result follows by noting that $\nR$ goes to infinity as $t$ goes to infinity.
	\end{proof}
	
	
	\subsection{UBMP Formulations}\label{subsec:singleHop:alphaUB}
	Using~\eqref{equivalentProb}, Lemma~\ref{lem:singleHopUB}, and $\mu_{2} = \infty$, we obtain the following UBMP problem.
	\begin{equation}\label{UBMP:singleHop}
	\begin{aligned}
	& \underset{\frac{1}{d} \leq R < \mu_1}{\text{min}} \quad \sum_{v=0}^{\infty} \Phi(v,R).
	\end{aligned}
	\end{equation}
	It is worth noting that the function $\Phi(0,R)$ is non-increasing in $R$ while the functions $\{\Phi(v,R):v > 1\}$ are non-decreasing in $R$. 
	
	A shortcoming of UBMP is that its objective function is intractable, in general, as it involves computation of a sum of infinite terms and each term requires computation of the distribution of sum of service times.
	To this end, we formulate Chernoff-UBMP obtained by using Chernoff bound for $\Phi(v,R)$ in Lemma~\ref{lem:singleHopUB}.
	
	\subsubsection{Chernoff-UBMP}
	Since $X^n_{1}$ are i.i.d, the Chernoff bound for $\Phi(v,R)$, defined in~\eqref{eq:Phi:singlehop}, is given by
	{\allowdisplaybreaks \begin{align}\label{eq:ChernoffPhi}
		\Phi(v,R) &\leq \min_{s \in \mathcal{S}}\; e^{-s(d+\frac{v-1}{R})} \mathbb{E}[e^{s \sum_{i=0}^{v} X_{1}^i}] \nonumber\\
		&= \min_{s \in \mathcal{S}}\; e^{-s(d+\frac{v-1}{R})} M^{v+1}_{1}(s) \nonumber \\
		&= \min_{s \in \mathcal{S}}\; e^{-s(d-\frac{1}{R})} M_{1}(s)  \beta^v_1(s),
		\end{align}
		where \jpcol{$\beta_1(s) = \frac{M_1(s)}{e^{s/R}}$ (defined in~\eqref{eq:beta}),} \jpcolor{and $\mathcal{S}$ is defined in~\eqref{eq:calS}}. Recall that, $\beta_1(s) < 1$ for all $s \in \mathcal{S}$. Therefore, using~\eqref{eq:ChernoffPhi} in the result of Lemma~\ref{lem:singleHopUB}, we obtain 
		\begin{align}\label{eq:ChernoffSinglehop}
		&\sum_{v=0}^{\infty} \Phi(v,R) \leq \sum_{v=0}^{\infty} \min_{s \in \mathcal{S}}\; e^{-s(d-\frac{1}{R})} M_{1}(s)\beta^v_1(s) \nonumber\\
		&\leq \min_{s \in \mathcal{S}}\; e^{-s(d-\frac{1}{R})} M_{1}(s) \sum_{v=0}^{\infty}  \beta^v_1(s)\nonumber \\
		& = \min_{s\in \mathcal{S}}\; \underbrace{e^{-s(d-\frac{1}{R})}\cdot \frac{M_{1}(s)}{(1 - \beta_1(s))}}_{\triangleq \, \Psi_1(s,d,R)}.
		\end{align}}
	Even though the Chernoff bound relaxes the upper bound in Lemma~\ref{lem:singleHopUB}, its objective function has a closed-form expression and can be computed numerically.
	The following theorem immediately follows from~\eqref{eq:ChernoffSinglehop} and Lemma~\ref{lem:singleHopUB}.
	\begin{theorem}\label{thm:singlehop}
		Given $d$, an upper bound for the violation probability for a single hop is given by
		\begin{align*}
		\lim_{t \rightarrow \infty} \P\{D(\nR) > t\} \leq \min_{s\in \mathcal{S}} \; \Psi_1(s,d,R),
		\end{align*}
		where $\Psi_1(s,d,R)$ is defined in~\eqref{eq:ChernoffSinglehop}.
	\end{theorem}
	With a slight abuse in the usage, we refer to the bound given in Theorem~\ref{thm:singlehop} as \textit{Chernoff bound}. 
	In the following, we formulate the Chernoff-UBMP for the single-hop scenario:
	\begin{equation}\label{Chernoff-UBMP:singleHop}
	\underset{ \frac{1}{d} \leq R < \mu_{1} }{\text{min}} \, \underset{s\in \mathcal{S}}{\text{min}} \; \Psi_1(s,d,R).
	\end{equation}
	
	\begin{lemma}\label{lem:singlehop:convexR}
		The function $\Psi_1(s,d,R)$ is strictly convex with respect to $\frac{1}{R}$.
	\end{lemma}
	\begin{proof}
		Recall that $T = \frac{1}{R}$. We prove that $\frac{\partial^2 \Psi_1(s,d,T)}{\partial T^2} > 0$ for all $s\in \mathcal{S}$. Let us define $f(T)$ as follows:
		\begin{align*}
		f(T) = \frac{e^{2sT}}{(e^{sT}-M_{1}(s))}.
		\end{align*} 
		Then, we rewrite $\Psi_1(s,d,T)$ as follows:
		\begin{align*}
		\Psi_1(s,d,T) = e^{-sd}[M_{1}(s)] f(T).
		\end{align*}
		From the above equation we infer that it is sufficient to prove $\frac{\partial^2 f(T)}{\partial T^2} > 0$. Taking first derivative $f'(T) = \frac{\partial f(T)}{\partial T}$, we obtain 
		\begin{align}\label{eq:firstder}
		f'(T) &= \frac{2se^{2sT}}{(e^{sT}-M_{1}(s))} - \frac{e^{2sT}2se^{sT}}{(e^{sT}-M_{1}(s))^2} \nonumber \\
		&= s f(T)\left[1-\frac{M_{1}(s)}{e^{sT}-M_{1}(s)}\right].
		\end{align}
		Taking the second derivative $f''(T) = \frac{\partial^2 f(T)}{\partial^2 T}$, we obtain 
		\begin{align*}
		&f''(T)\! =\! sf'(T) \left[1 \! -\frac{M_{1}(s)}{e^{sT}-M_{1}(s)}\right]\!\! +\! \frac{s^2 f(T)M_{1}(s)e^{sT}}{(e^{sT}-M_{1}(s))^2} \\
		&= \! s^2\! f(T) \left[1\! -\frac{M_{1}(s)}{e^{sT}-M_{1}(s)}\right]^2\!\! +\! \frac{s^2 f(T)M_{1}(s)e^{sT}}{(e^{sT}-M_{1}(s))^2} > 0.
		\end{align*}
		In the second step above we have used~\eqref{eq:firstder}. The last step follows by noting that $e^{sT} > M_{1}(s)$ for all $s\in \mathcal{S}$, $M_{1}(s)>0$ for all $s$, and $f(T) > 0$. 
	\end{proof}
	
	\begin{lemma}\label{lem:singlehop:convexs}
		For $s > 0$, the function $\Psi_1(s,d,R)$ is convex in $s$.
	\end{lemma}
	\begin{proof}
		We have
		\begin{align*}
		&\Psi_1(s,d,R) =\frac{e^{-s(d-\frac{1}{R})} M_{1}(s)}{(1 - \beta_1(s))}\\
		& = e^{-s(d-\frac{1}{R})} \sum_{v=0}^{\infty}  M_{1}(s) \beta^v_1(s)\\
		&= \sum_{v=0}^{\infty} e^{-s(d+\frac{v-1}{R})} M^{v+1}_{1}(s) = \sum_{v=0}^{\infty} (\mathbb{E}[e^{-s\hat{X}}])^{v+1},
		\end{align*}
		where $\hat{X} = (d+\frac{v-1}{R})/(v+1)-X^1_1$. Recall that the sum of convex functions is a convex function. Therefore, from the last step above, we infer that $\Psi_1(s,d,R)$ is convex if $(\mathbb{E}[e^{-sX}])^{v+1}$ is convex for $v \in \{0,1,\ldots$\}. For $s>0$, $e^{-s\hat{X}}$ is convex in $s$ for any $v$ and and any realization of $X^1_1$. Therefore, $\mathbb{E}[e^{-s\hat{X}}]$ is convex, and since $x^{v+1}$ is convex and increasing in $x$, we have that $(\mathbb{E}[e^{-sX}])^{v+1}$ is convex. Hence the result is proven.
	\end{proof}
	Both Lemmas~\ref{lem:singlehop:convexR} and~\ref{lem:singlehop:convexs} can be leveraged to efficiently solve~\eqref{Chernoff-UBMP:singleHop}.
	The heuristic solutions we obtain by solving the Chernoff-UBMP can be improved further for service distributions for which the distribution of a finite sum of service times can be computed exactly. Therefore, we next propose a relatively tight upper bound called $\alpha$-relaxed upper bound and formulate $\alpha$-UBMP.
	
	\subsubsection{$\alpha$-UBMP}
	In the upper bound provided in Lemma~\ref{lem:singleHopUB}, we propose to compute first $K < \infty$ terms of the summation, and use Chernoff bound for the rest of the terms. In the following, we make this precise. We first present a bound on the summation starting from $K$.
	\begin{lemma}\label{lem:singlehop:alpharelaxed}
		For any $K \geq 0$, we have
		\begin{align*}
		\sum_{v = K}^{\infty} \Phi(v,R) \leq \min_{s \in \mathcal{S}} \Psi_1(s,d,R)\beta^K_1(s).
		\end{align*}
	\end{lemma}
	\begin{proof}
		The result follows by using the upper bound for $\Phi(v,R)$ given in~\eqref{eq:ChernoffPhi} and repeating the steps in~\eqref{eq:ChernoffSinglehop} for the summation over $v$ from $K$ to infinity.
	\end{proof}
	For the single hop scenario we define $\alpha$ as follows.
	\begin{align*}
	\alpha = 1 + \frac{\min_{s \in \mathcal{S}} \Psi_1(s,d,R)\beta^K_1(s)}{\sum_{v = 0}^{K-1} \Phi(v,R)}.
	\end{align*}
	Note that $\alpha$ depends on the value of $K$. Using Lemmas~\ref{lem:singleHopUB} and~\ref{lem:singlehop:alpharelaxed}, we next state the $\alpha$-relaxed upper bound without proof.
	\begin{theorem}
		Given $d$, the $\alpha$-relaxed upper bound for the violation probability for a single hop is given by
		\begin{align*}
		\lim_{t \rightarrow \infty}\! \P\{\!D(\nR)\! >\! t\}\! \leq\!\! \sum_{v = 0}^{K-1}\!\! \Phi(v,R)\! +\! \min_{s\in \mathcal{S}} \! \Psi_1(s,d,R)\beta^K_1\!(s).
		\end{align*}
	\end{theorem} 
	Note that, by definition the $\alpha$-relaxed upper bound is at most $\alpha$ times worse than the upper bound $\sum_{v = 0}^{\infty} \Phi(v,R)$. More precisely, the $\alpha$-relaxed upper bound has $\alpha$ approximation factor with respect to $\sum_{v = 0}^{\infty} \Phi(v,R)$. To see this,
	{\allowdisplaybreaks\begin{align*}
	&\sum_{v = 0}^{K-1} \Phi(v,R) + \min_{s\in \mathcal{S}} \; \Psi_1(s,d,R)\beta^K_1(s)\! \\
	&= \! \sum_{v = 0}^{K-1} \Phi(v,R)\!\left(1 + \frac{\min_{s \in \mathcal{S}} \Psi_1(s,d,R)\beta^K_1(s)}{\sum_{v = 0}^{K-1} \Phi(v,R)} \right) \\
	&\leq \! \alpha \! \sum_{v = 0}^{\infty} \Phi(v,R).
	\end{align*}}
	Note that $\alpha > 1$, and it is easy to see that as $K$ increases, the value of $\alpha$ approaches $1$ from above. In this work, we choose $K$ the largest value that is computationally tractable in numerical evaluations. 
	Now, we formulate \prob~as follows:
	\begin{equation*}
	\begin{aligned}
	\underset{\frac{1}{d} \leq R < \mu_{1}}{\text{min}} \sum_{v = 0}^{K-1} \Phi(v,R) + \min_{s\in \mathcal{S}} \Psi_1(s,d,R)\beta^K_1(s).
	\end{aligned}
	\end{equation*} 

	\section{Extensions to Two-Hop and N-Hop Scenarios}\label{sec:twohop}
	In this section, we present Chernoff-UBMP and \prob~ for the two-hop scenario and also present Chernoff-UBMP for N-hop tandem queuing network. 
	
	In the following we first focus on the two-hop scenario. Similar to the case of single-hop scenario, we use Reich's equation and apply union bound to obtain an upper bound for the AoI violation probability which is presented in the following lemma.  
	\begin{lemma}\label{lem:twoHop}
		Given $d$, and $\nR$ as defined in \eqref{eq:nR}, we have
		\begin{align*}
		\lim_{t\rightarrow \infty} \P(D(\nR) > t)\! \leq \! \lim_{\nR\rightarrow \infty} \sum_{v_0=0}^{\nR} \sum_{v_1=0}^{\nR-v_0} \Phi(v_0,v_1,R),
		\end{align*}
		where
		\begin{align*}
		\Phi(v_0,v_1,R) \triangleq \P\left\{\!\sum_{i=0}^{v_0}\! X_{2}^{i} +\! \sum_{i=0}^{v_1}\! X_{1}^{i} \! >\! d\! +\! \frac{v_0 + v_1 - 1}{R}\!\right\}.
		\end{align*}
	\end{lemma}
	\begin{proof}
		The proof is given in Appendix~\ref{lem:twoHop:proof}.
	\end{proof}
	
	\subsection{Chernoff-UBMP for Two-Hop Scenario}
	\begin{theorem}\label{thm:twoHop:Chernoff}
		For the two-hop network with deterministic arrivals, the violation probability is upper bounded as follows:
		\begin{align*}
		\lim_{t \rightarrow \infty} \P\{D(\nR) > t\} \leq \min_{s\in \mathcal{S}} \Psi_2(s,d,R), 
		\end{align*}
		where
		\begin{align}\label{eq:Psi2}
		\Psi_2(s,d,R) =  \frac{e^{-s(d-\frac{1}{R})} M_{1}(s) M_{2}(s)}{(1 - \beta_1(s))(1 - \beta_2(s))}.
		\end{align}
	\end{theorem}
	\begin{proof}
		We use the relation between departure times, arrival times and the service times given by~\eqref{eq:departureTime} iteratively and apply union bound and Chernoff bound to obtain the result. The details of the proof are given in Appendix~\ref{thm:twoHop:Chernoff:proof}.
	\end{proof}
	
	
	The Chernoff-UBMP problem for the two-hop network is stated below:
	\begin{equation}\label{prob:twoHopUBMP}
	\begin{aligned}
	\underset{\frac{1}{d} \leq R < \mu}{\text{min}}\, \min_{s \in \mathcal{S}}\; \Psi_2(s,d,R).
	\end{aligned}
	\end{equation}
	The lemmas below provide convexity properties of $\Psi_2(s,d,R)$. Since the proofs of the lemmas are similar to that in the case of single-hop scenario (Lemmas~\ref{lem:singlehop:convexR} and~\ref{lem:singlehop:convexs}), we omit them here.
	\begin{lemma}\label{lem:twHopConvexR}
		For the two-hop network with deterministic arrivals, given $s \in \mathcal{S}$ and $d > 0$, $\Psi_2(s,d,R)$ is convex with respect to $\frac{1}{R}$.
	\end{lemma}
	
	\begin{lemma}\label{lem:twHopConvexs}
		For the two-hop network with deterministic arrivals, given $s \in \mathcal{S}$ and $d > 0$, $\Psi_2(s,d,R)$ is convex with respect to $s$.
	\end{lemma}
	
	\subsection{$\alpha$-UBMP for Two-Hop Scenario}
	In the following theorem we present the $\alpha$-relaxed upper bound.
	\begin{theorem}\label{thm:twohop:alphaUB}
		For the two-hop network with deterministic arrivals, for any $K \geq 1$, the $\alpha$-relaxed upper bounded is given by
		\begin{align*}
		\sum_{v_0=0}^{K - 1}  \sum_{v_{1}=0}^{K-1} \Phi(v_0,v_1,R) + \min_{s\in \mathcal{S}} \Psi(s,d,R,K),
		\end{align*}
		where
		\begin{align*}
		&\Psi(s,d,R,K) \nonumber\\
		&= \! e^{-s(d-\frac{1}{R}\!)} M_{1}(s)M_{2}(s) \frac{(\beta_1^K\!(s)\! +\! \beta_2^K\!(s)\!-\!\beta_1^K\!(s)\beta^K_2\!(s)\!)}{(1\!-\!\beta_1(s))(1\!-\!\beta_2(s))}.
		\end{align*}
	\end{theorem}
	\begin{proof}
		The proof is given in Appendix~\ref{thm:twohop:alphaUB:proof}.
	\end{proof}
	We note that the $\alpha$-relaxed upper bound is computationally expensive when compared to that in the single-hop scenario because of the nested sum.

	{\color{black}\subsection{N-hop Scenario}\label{sec:multihop}}
	For an N-hop tandem network we have $k \in \{1,2,\ldots,N\}$ and $D(n) = D_N(n)$. For simplicity of presentation, in this section, we assume that $X^n_{k}$ are identically distributed. Therefore, we have $\mu = \mu_{k}$ for all $k$, and $M_{k}(s) = M_{1}(s)$ for all k. We now define the set $\mathcal{S}$ as follows.
	\begin{align*}
	\mathcal{S} = \{s > 0: M_{1}(s) < e^{s/R}\}.
	\end{align*}
	
	\begin{lemma}\label{lem:NHop}
		Given $d$, and $\nR$ as defined in \eqref{eq:nR}, we have
		\vspace{-.3cm}
		\begin{align*}
		\lim_{t\rightarrow \infty} \P(D(\nR)\! > \! t) \! \leq \! \lim_{\nR\rightarrow \infty}\!\! \sum_{v_0=0}^{\nR}\!\! \sum_{v_1=0}^{\nR-v_0}\!\!\!\ldots \!\!\!\!\sum_{v_{N-1}=0}^{\nR-v_{N-2}}  \!\!\!\!\Phi(v_0^{N-1}\!\!,\!R),
		\vspace{-.5cm}
		\end{align*}
		where
		\vspace{-.3cm}
		\begin{align}\label{eq:Phi:Nhop}
		\Phi(v_0^{N-1}\!,R)\! \triangleq \! \P\left\{\!\sum_{k=0}^{N-1}\! \sum_{i=0}^{v_k} X_{N-k}^{i} > d \!+\! \frac{\sum_{k=0}^{N-1} v_k\! -\! 1}{R}\!\right\},
		\end{align}
		and $v_0^{N-1} = (v_0,v_1,\ldots,v_{N-1})$.
	\end{lemma}
	\begin{proof}
		The proof follows similar steps as the proof of Lemma~\ref{lem:twoHop} and is omitted.
	\end{proof}
	\begin{theorem}\label{thm:multihop}
		For the $N$-hop network with deterministic arrivals, the violation probability is upper bounded as follows:
		\begin{align*}
		\lim_{t \rightarrow \infty} \P\{D(\nR) > t\} \leq \min_{s\in \mathcal{S}} \Psi_N(s,d,R), 
		\end{align*}
		where
		\begin{align}\label{eq:Psi}
		\Psi_N(s,d,R) =  \frac{e^{-s(d-\frac{1}{R})}[M_{1}(s)]^{N}}{[1 - \frac{M_{1}(s)}{e^{s/R}}]^{N}}.
		\end{align}
	\end{theorem}
	\begin{proof}
		We use the relation between departure times, arrival times and the service times given by~\eqref{eq:departureTime} recursively starting from the last node $N$, and apply union bound and Chernoff bound to obtain the result. The proof follows similar steps as in the proof of Theorem~\ref{thm:twoHop:Chernoff} and therefore it is omitted.
	\end{proof}
	Therefore, an upper bound minimization problem for the N-hop network can be stated as follows:
	\begin{equation}\label{prob:NHopUBMP}
	\underset{\frac{1}{d} \leq R < \mu}{\text{min}} \, \underset{s \in \mathcal{S}}{\text{min}} \, \Psi_N(s,d,R).
	\end{equation}
	
	\noindent \textbf{\textit{Discussion:}} We note that similar to the single-hop and two-hop scenario $\Psi_N(s,d,R)$ is also convex with respect to $\frac{1}{R}$ and with respect to $s$. One may also obtain $\alpha$-UBMP for the $N$-hop scenario. However, the $\alpha$-relaxed upper bound involves the nested sum which becomes computationally expensive as $N$ increases. Furthermore, we note that as $N$ increases the upper bounds become more relaxed and therefore the heuristic solutions provided by Chernoff-UBMP may not be close to optimal solution. Nevertheless, these heuristic solutions could potentially be used as starting points. {\color{black}For example, when the controller has non-negligible processing time, the sensor-controller-actuator can be modelled as a three-hop tandem queuing system and one may use the heuristic solutions provided by the Chernoff-UMBP for three-hop scenario}.
	
	Next, we present an independent result for service-time distributions with bounded support. 
	
	\subsubsection*{Service Distributions with Bounded Support}
	Note that in practice, the service time distributions typically have bounded support. For example, the channel capacity for transmissions is always upper bounded due to bandwidth limitation. Considering that the service time is upper bounded by $b \in \mathbb{R}_{>0}$, in the following theorem we present a result for computing an optimal rate for age limits above certain threshold.
	\jpcolor{\begin{corollary}
			For an N-hop network, if the support of the service time distribution is upper bounded by $b < \infty$, then for all $d \geq (N+1)b$, the AoI violation probability is zero at \jpcolor{$R \leq (N+1)/d$}, i.e., these rate solutions are optimal for~\eqref{equivalentProb}.
		\end{corollary}
		\begin{proof}
			We rewrite $\Phi(v_0^{N-1},R)$ (defined in~\eqref{eq:Phi:Nhop}) as follows:
			\begin{align*}
			\Phi(v_0^{N-1}\!,R)\! =\! \P\left\{\!\sum_{k=0}^{N-1}\!\sum_{i=0}^{v_k}\! \left(\!X_{N-k}^{i} - \frac{1}{R}\right)\! > \!d \! -\! \frac{N+1}{R}\!\right\}.
			\end{align*}
			For $R \leq (N+1)/d$, we have $X^n_k \leq b \leq \frac{1}{R^*}$ for all $k \geq 1$ and for all $n$, and we obtain
			\begin{align*}
			\Phi(v_0^{N-1}\!,R^*) &= \P\left\{\sum_{k=0}^{N-1}\sum_{i=0}^{v_k} \left(X_{N-k}^{i} - \frac{1}{R}\right) > 0\right\}\! = 0.
			\end{align*}
			Therefore, from Lemma~\ref{lem:NHop} we conclude that the AoI violation probability $\lim_{t\rightarrow \infty} \P(T_D(\hat{n}_{R}) > t)$ is equal to zero when $R \leq (N+1)/d$. 
		\end{proof}
	}

	\section{Application Examples: Geometric, Exponential and Erlang Service}\label{sec:exampleDis}
	In the following we show the computation of the upper bounds for typical service distributions, namely, geometric, exponential and Erlang. These distributions are most commonly used in the queuing analysis, and also they serve as good models for several practical service-time processes. 
	Note that for these distributions, the distribution of the sum of service times is known and thus the $\alpha$-relaxed upper bound can be computed. Later in Section~\ref{sec:numerical} we will evaluate the performance of the the computed heuristic solutions for these service distributions. To shorten the expressions, in the sequel we denote 
	\begin{align*}
	Y_1 = \sum_{i=0}^{v_{1}} X_{1}^{i}, \text{ }
	Y_2 = \sum_{i=0}^{v_0} X_{2}^{i}, \text{ and }
	\kappa  = d + \frac{v_0 + v_1 - 1}{R}.
	\end{align*}

	\subsection{Geometric Service: Wireless Links with Packet Errors}
	Consider that each packet generated by the sensor is of fixed length and the packets that carry actuator commands are also of fixed length, possibly different from sensor packet length. To accommodate for packet transmission errors in the wireless links, we use
	geometric distribution to model the number of time slots required for transmitting a packet successfully. In particular, we consider the service distributions at link $1$ and link $2$ to be geometric with success probabilities $p_1$ and $p_2$, respectively. Given an age limit $d$ at the actuator, we compute $R$ heuristically.  
	
	In the following we compute the first term of the $\alpha$-relaxed upper bound given in Theorem~\ref{thm:twohop:alphaUB}. Since $Y_1$ and $Y_2$ are integers, we have
	\begin{align}\label{eq:probterm}
	&\sum_{v_0=0}^{K - 1}  \sum_{v_{1}=0}^{K-1} \Phi(v_0,v_1,R) = \sum_{v_0=0}^{K - 1}  \sum_{v_{1}=0}^{K-1} \P \left\{Y_1  +  Y_2 \! > \kappa \!\right\} \nonumber\\
	&= \sum_{v_0=0}^{K - 1}  \sum_{v_{1}=0}^{K-1} \P \left\{Y_1  +  Y_2 \! > \lfloor \kappa \rfloor\right\}.
	\end{align}
	Since for geometrical distribution $X^{i}_k \geq 1$, for all $i$ and $k \in \{1,2\}$, we have $Y_1 \geq v_1 + 1$ and $Y_2 \geq v_0 + 1$. Therefore, for $\lfloor \kappa \rfloor <= v_1 + v_2 + 1$, we have $\P \{Y_1  +  Y_2 > \lfloor \kappa \rfloor\} = 1$. For $\lfloor \kappa \rfloor >= v_1 + v_2 + 2$ we compute the probability by conditioning on $Y_2=y$ for positive integers $y \geq v_0 + 1$.
	\begin{align*}
	& \P \left\{Y_1 \! + \!  Y_2 \! > \! \lfloor \kappa \rfloor \right\}\\
	& = \!\! \sum_{y = v_0 + 1}^{ \infty} \!\! \P \left\{Y_1 \! + \!  Y_2 \! > \! \lfloor \kappa \rfloor  | Y_2\! =\! y \right\}\!\P\{Y_2 \!=\! y\} \\
	&=\!\! \sum_{y = v_0 + 1}^{\lfloor \kappa \rfloor - v_1 - 1}\!\!\!\!\!\! \P \left\{Y_1 \! >\!  \lfloor \kappa \rfloor \! - \! y  \right\}\!\P\{Y_2 = y\}\! + \! \P\{Y_2 \! \geq \! \lfloor \kappa \rfloor \! - \! v_1\}.
	\end{align*}
	In the last step above we have used $\P \left\{Y_1  >  \lfloor \kappa \rfloor - y  \right\} = 1$ for $y \geq \lfloor \kappa \rfloor - v_1$.
	Noting that the sum of i.i.d. geometric random variables has a negative binomial distribution, we have 
	\begin{align*}
	\P\{Y_2 = y\} &= \P\left\{\sum_{i=0}^{v_0} X_{2}^{i} = y\right\} \\
	&  = \binom{y-1}{v_0}p_2^{v_0+1}(1-p_2)^{y-v_0-1},
	\end{align*}
	and
	\begin{align*}
	\P \left\{Y_1 \! >  \lfloor \kappa \rfloor  - y \right\} = \frac{B(1-p_2; \lfloor \kappa \rfloor  - y - v_1,v_1 + 1)}{B(\lfloor \kappa \rfloor  - y - v_1,v_1 + 1)},
	\end{align*}
	where $B(\cdot)$ is the incomplete beta function given by
	\begin{align*}
	B(z;a,b) &= \int_{0}^{z} x^{a} (1-x)^{b} dx, \\
	B(a,b) &= \int_{0}^{1} x^{a} (1-x)^{b} dx .
	\end{align*}
	Similarly, we compute $\P\{Y_2 \geq \lfloor \kappa \rfloor -v_1\}$. Finally, using $\P \left\{Y_1  +  Y_2 \! > \lfloor \kappa \rfloor \right\}$ we compute~\eqref{eq:probterm}. For computing the Chernoff bound we require the moment generating function, which for geometric service is given below.
	\begin{align*}
	M_{k}(s) = \frac{p_k e^s}{1-(1-p_k)e^s}.
	\end{align*}
	Since the Chernoff bound is convex in $s$, using bisection algorithm we compute its minimum value.

	\subsection{Exponential Service}
	In this subsection, we study the two-hop system with exponentially distributed service times with rates $\mu_{1}$ and $\mu_{2}$ at links $1$ and $2$, respectively. For this case, $Y_1$ is a sum of $v_1+1$ i.i.d. exponential random variables, which is given by the Gamma distribution with shape parameter $v_1+1$ and rate parameter $\mu_{1}$. Similarly, $Y_2$ has Gamma distribution with shape parameter $v_2+1$ and rate parameter $\mu_{2}$.  Therefore, we compute  $\Phi(v_0,v_1,R)$ as follows.
	\jpcol{\begin{align}\label{eq:phiExp}
		&\Phi(v_0,v_1,R) = \int_{0}^{\infty} \P\{Y_1 > \kappa - y\}f_{Y_2}(y) dy \nonumber \\
		&= \int_{0}^{\kappa} \P\{Y_1 > \kappa - y\}f_{Y_2}(y) dy + \int_{\kappa}^{\infty} f_{Y_2}(y) dy,
		\end{align}}
	where $f_{Y_2}(\cdot)$ is the PDF of $Y_2$, given by
	\begin{align*}
	f_{Y_2}(y) &= \frac{\mu^{v_2+1}_{2} y^{v_2} e^{-\mu_{2} y}}{v_2!} \, ,\\
	\P\{Y_1 > \kappa - y\} &= \frac{\Gamma(v_1+1,\mu_{1}(\kappa - y))}{v_1!},
	\end{align*}
	and $\Gamma(x,a)$ is the upper incomplete gamma function:
	\begin{align*}
	\Gamma(x,a) = \int_{a}^{\infty}y^{x-1}e^{-y} dy.
	\end{align*}
	Further, if $\mu_1 = \mu_2 = \mu$, then 
	\begin{align*}
	\Phi(v_0,v_1,R)\big|_{\mu_1 = \mu_2} = \frac{\Gamma(v_0+v_1+2,\mu\kappa)}{(v_0+v_1+1)!}.
	\end{align*} 
	
	For computing the Chernoff bound we use the MGF of the exponential distribution which is given below.
	\begin{align*}
	M_{k}(s) = \frac{\mu_{k}}{\mu_{k} - s}, \text{ for }s < \mu_{k}.
	\end{align*}
	
	\subsection{Erlang Service}
	Consider the Erlang service distribution at link $k$ has shape parameter $b_k$ and rate $\lambda_k$. This implies $\mu_k = b_k \lambda_k$. We note that, in this case, $Y_k$ has Gamma distribution with shape parameter $(v_k+1)b_k$ and rate parameter $\lambda_k$. Therefore, we compute the bounds using similar expressions given in the previous subsection.
	
	\textit{\textbf{Remark 2:}} We note that the Chernoff upper bound and the $\alpha$-relaxed upper bound presented above may take values greater than $1$. It is natural to cap the values of these upper bounds by $1$ because for probability values an upper bound greater than $1$ is not of any use, in general. However, somewhat to our surprise, in our simulations we found that allowing the values of the proposed bounds greater than $1$ provides good heuristic solutions for the sampling rate, especially for parameter setting where the upper bounds are always greater than $1$. Since our primary objective is to find upper bounds that can provide good heuristic solutions, but need not necessarily be tight upper bounds, we consider values greater than $1$ for the bounds in our numerical evaluation. However, this should not be confused with the violation probability which does not exceed 1 at all times.

	\section{Numerical Evaluation}\label{sec:numerical}
	In this section, we evaluate the performance of \prob~solutions and \Chernoff~solutions for geometric, exponential and Erlang service distributions. We first study the trends of the proposed upper bounds in comparison to the AoI violation probability obtained using simulation for both single-hop and two-hop scenarios. We then evaluate the quality of numerically computed solutions using the UBMPs in comparison with that of the simulation-based estimate of the optimum violation probability. \jpcol{Finally, we present simulation results comparing the performance of FCFS with queue management policies that use unit buffer and packet replacement.}
	
	\jpcolor{Since Chernoff-UBMP is a convex optimization problem we used bisection search, and for $\alpha$-UBMP we used brute-force search to compute the respective optimal rates}. The numerical computations are done using MATLAB, and the simulation is implemented in C where we run $10^{10}$ iterations for each data point. The default parameters are as follows. For exponential distribution $\mu_1$ and $\mu_2$ equal 1 packet/ms; for Erlang distribution we use shape parameters $b_1 = b_2 = 3$ and rate parameters $\lambda_1 = \lambda_2 = 3$, and therefore the mean rates $\mu_1$ and $\mu_2$ equal one packet/ms;  for geometric service we choose success probabilities $p_1 = 0.85$ and  $p_2 = 0.9$, \jpcol{and the slot duration is $1$ ms}. The minimum value for $R$ is chosen to be $0.2$ packets/ms and its maximum value is chosen to be $0.75*\min(\mu_1,\mu_2)$ packets/ms. \jpcol{For all the figures with varying rate $R$ on the x-axis, a constant resolution of $0.025$ is used.} The minimum value for $d$ is chosen to be $5$ ms and its maximum value is chosen to be $15$ ms. 
	We use $K = 30$ for computing $\alpha$-relaxed upper bound for all the distributions because for Geometric service MATLAB does not provide precision guarantees for higher $K$ values for computing $\Phi(v_0,v_1,R)$, and for other service distributions, choosing $K = 30$ is sufficient to obtain $\alpha$ values close to $1$. 
	
	\begin{figure*}[ht!]
		\centering
		\begin{subfigure}[b]{0.30\textwidth}
			\includegraphics[width=\textwidth]{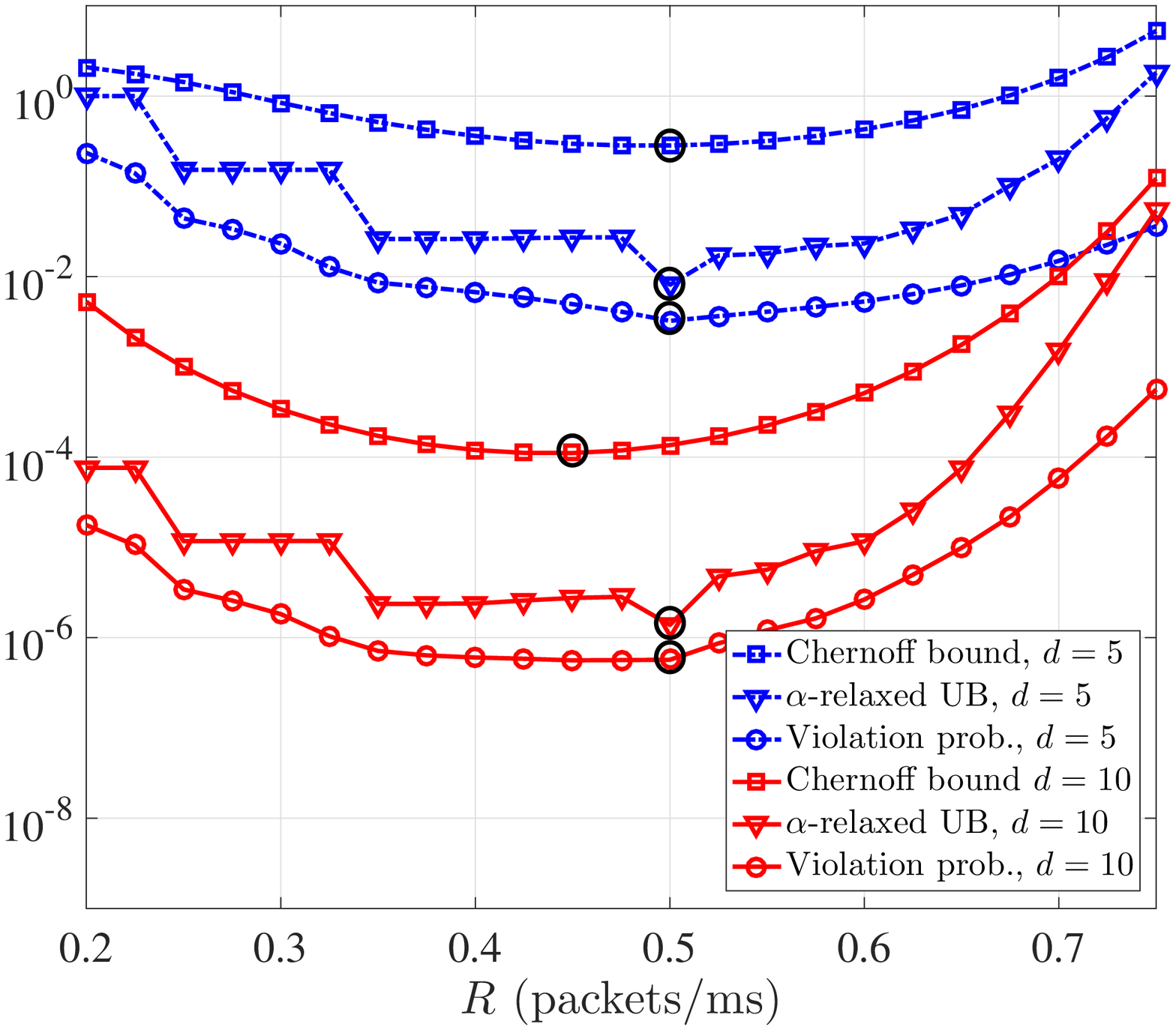}
			\caption{Geometric service with success probability $p_1 = 0.85$.}
			\vspace{-0.0cm}
			\label{fig:singleHop_Geometric_varR}
		\end{subfigure} \quad 
		\begin{subfigure}[b]{0.30\textwidth}
			\includegraphics[width=\textwidth]{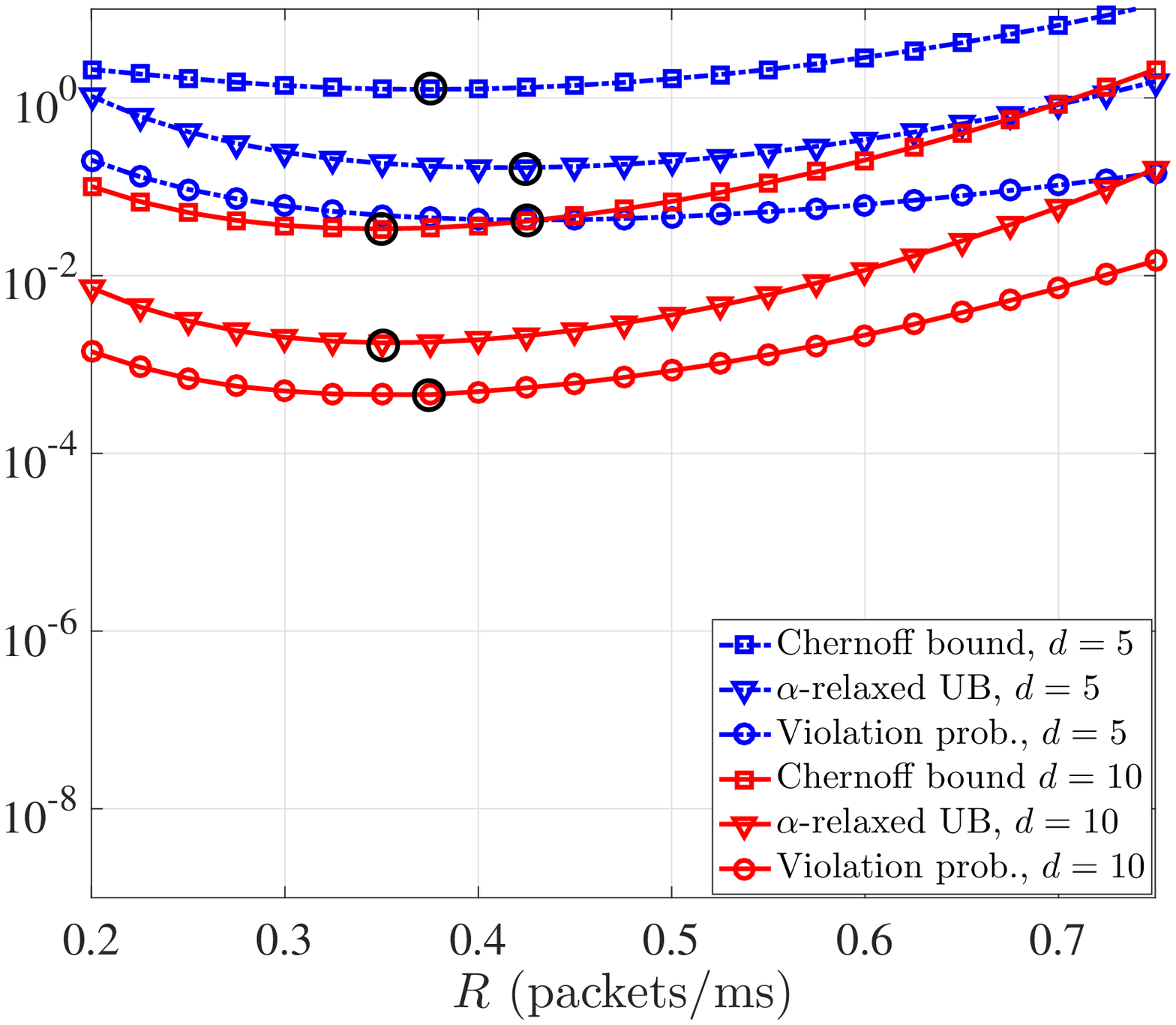}
			\caption{Exponential service with mean rate $\mu_1 = 1$.}
			\vspace{-0.0cm}
			\label{fig:singleHop_Exp_varR}
		\end{subfigure} \quad 
		\begin{subfigure}[b]{0.30\textwidth}
			\includegraphics[width=\textwidth]{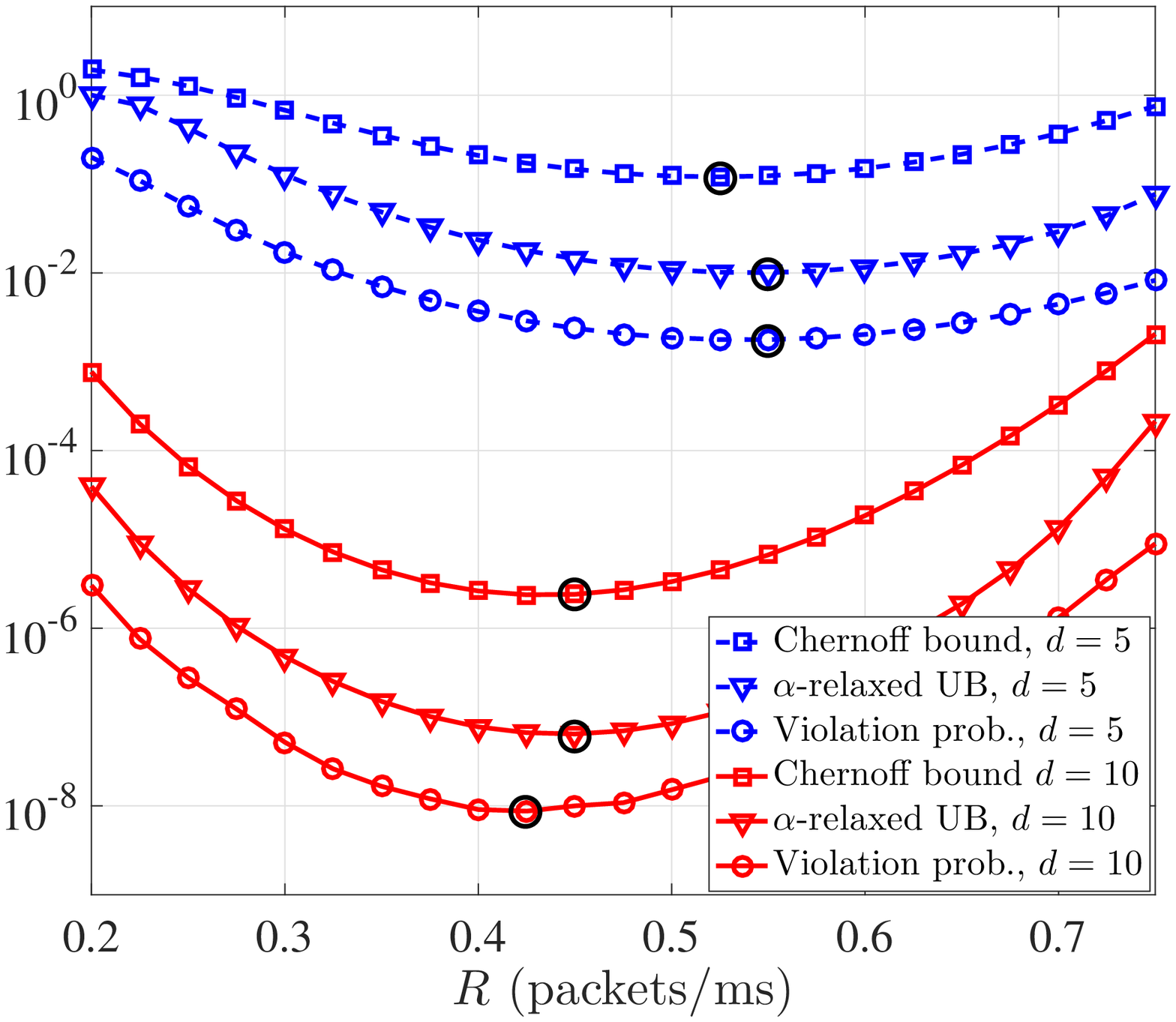}
			\caption{Erlang service with $b_1 = 3$, $\lambda_1 = 3$ and $\mu_1 = 1$.}
			\vspace{-0.0cm}
			\label{fig:singleHop_Erlang3_varR}
		\end{subfigure}
		\caption{Comparison of the upper bounds for varying arrival rate $R$ in a \textit{single hop} for different service time distributions.}\label{fig:singleHop_varR}
	\end{figure*}
	
	\begin{figure*}[ht!]
		\centering
		\begin{subfigure}[b]{0.30\textwidth}
			\includegraphics[width=\textwidth]{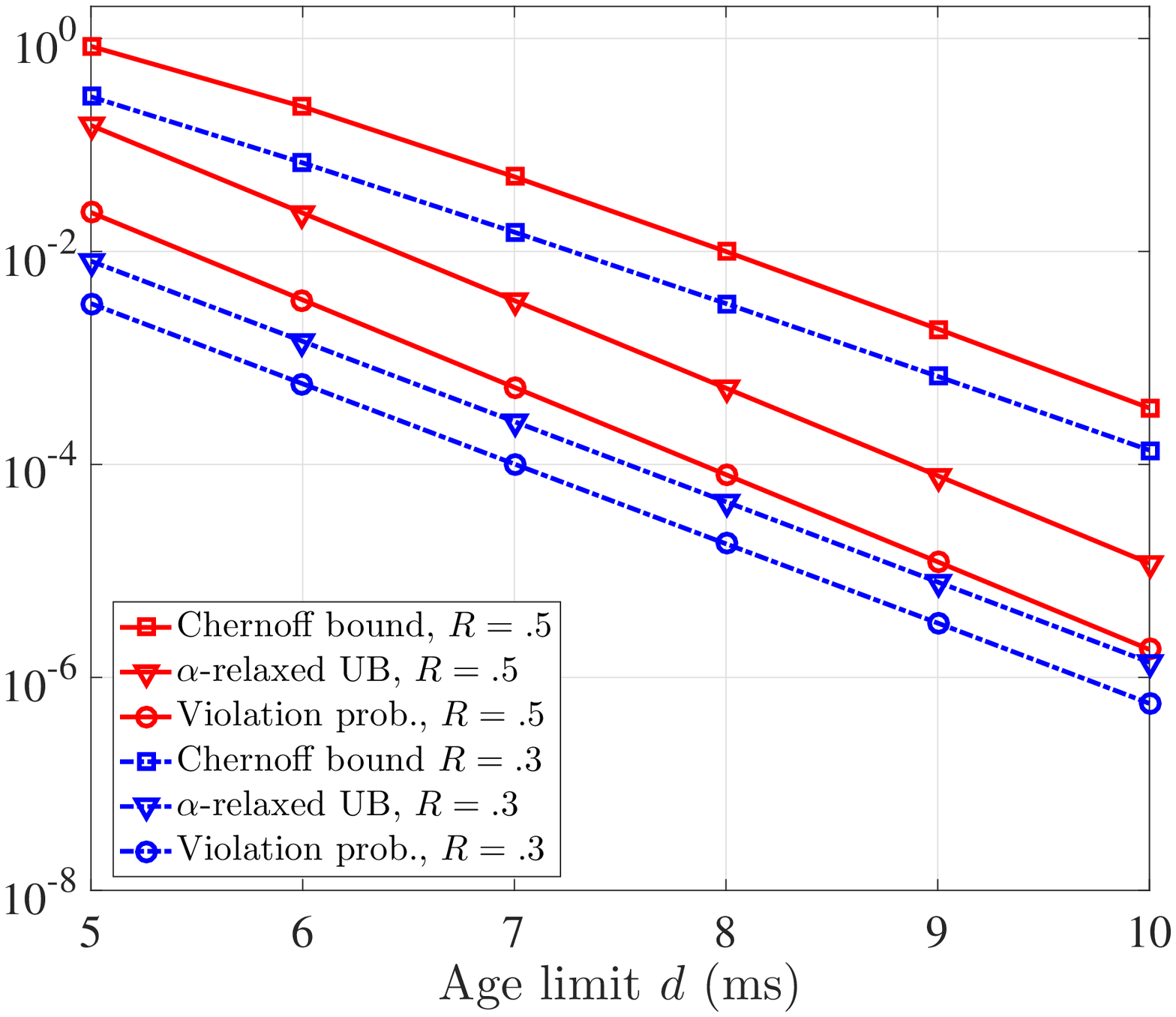}
			\caption{Geometric service with success probability $p_1 = 0.85$.}
			\vspace{-0.0cm}
			\label{fig:singleHop_Geometric_vard} 
		\end{subfigure} \quad 
		\begin{subfigure}[b]{0.30\textwidth}
			\includegraphics[width=\textwidth]{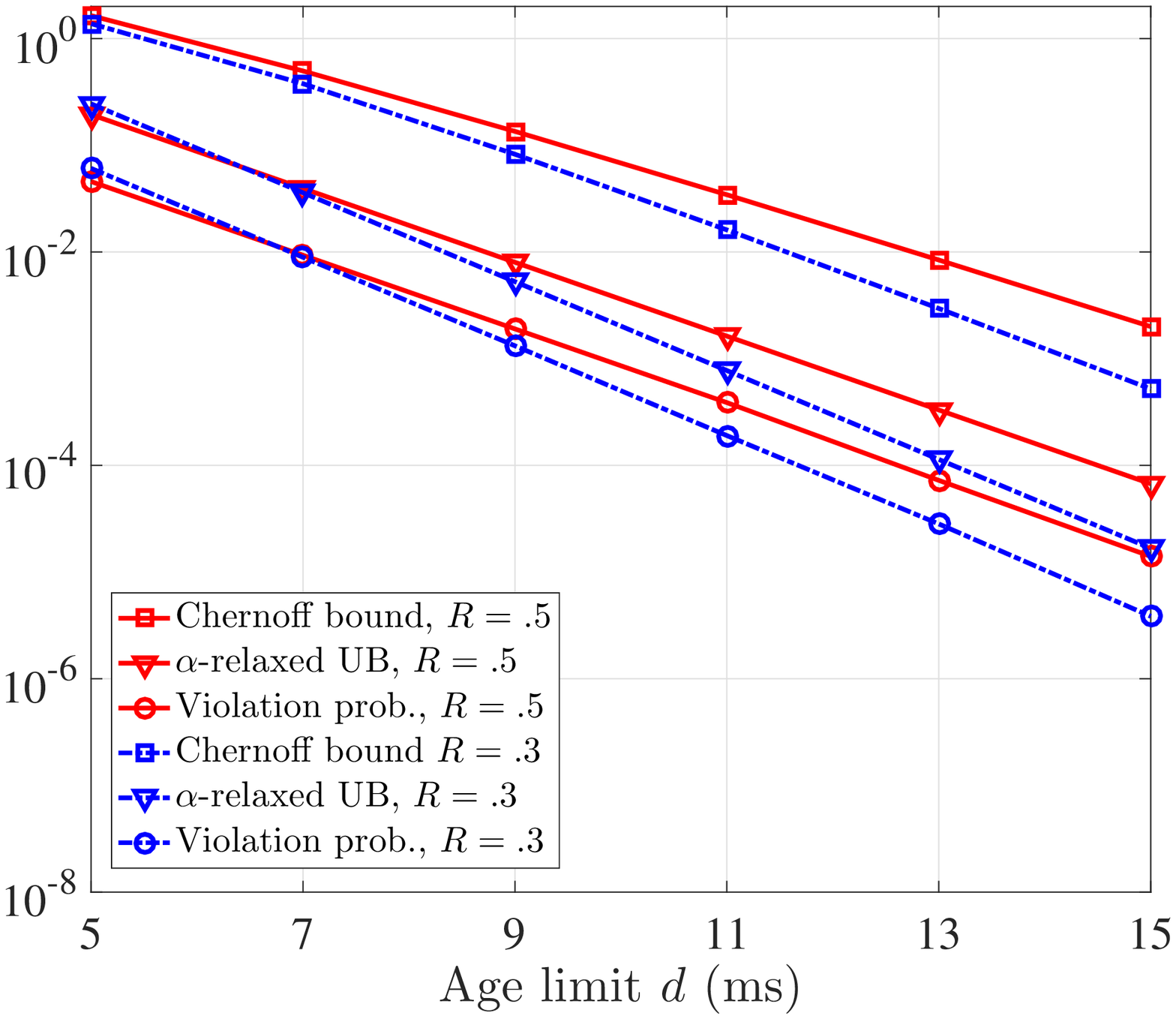}
			\caption{Exponential service with mean rate $\mu_1 = 1$.}
			\vspace{-0.0cm}
			\label{fig:singleHop_Exp_vard} 
		\end{subfigure} \quad 
		\begin{subfigure}[b]{0.30\textwidth}
			\includegraphics[width=\textwidth]{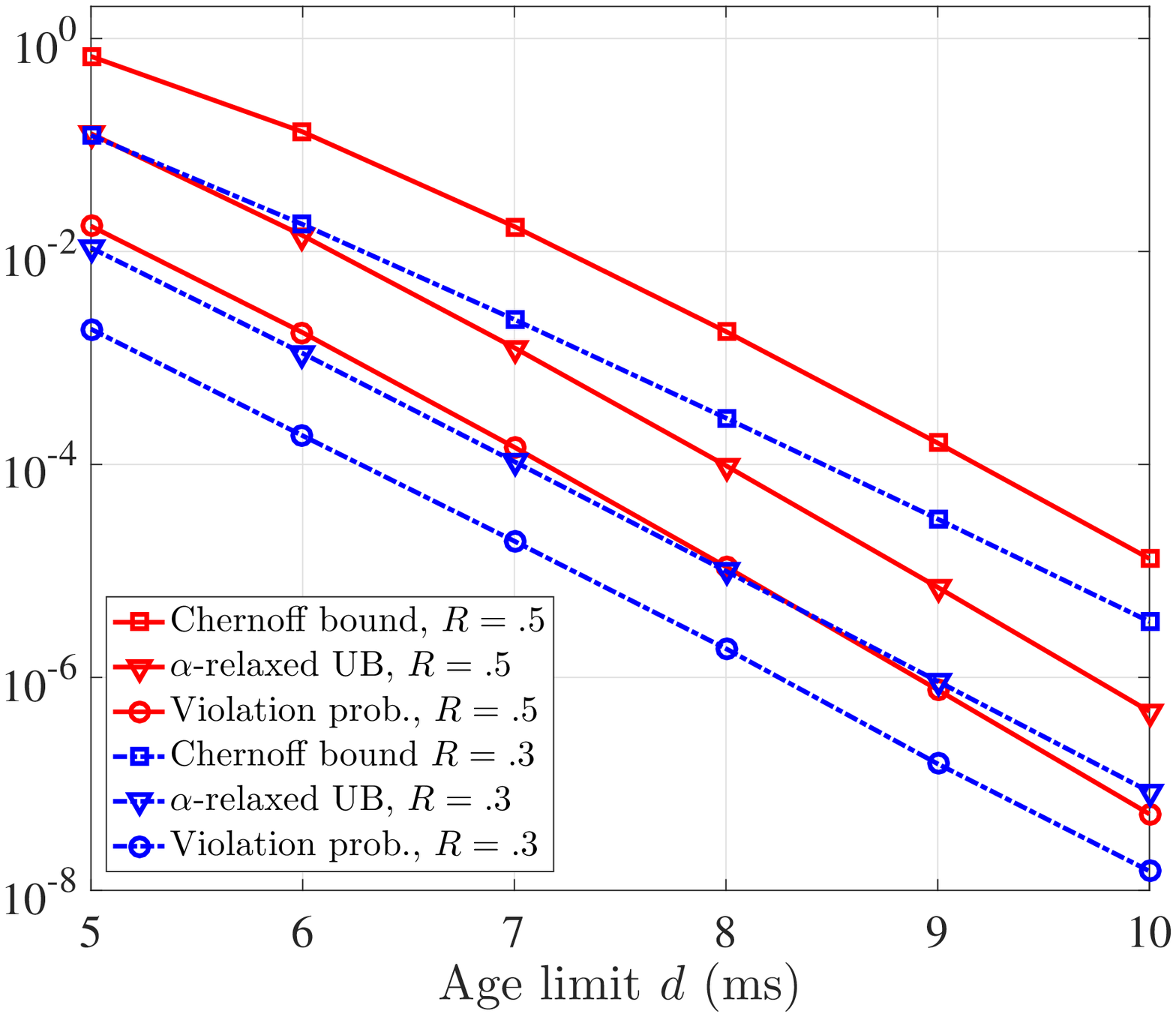}
			\caption{Erlang service with $b_1 = 3$, $\lambda_1 = 3$ and $\mu_1 = 1$.}
			\vspace{-0.0cm}
			\label{fig:singleHop_Erlang3_vard} 
		\end{subfigure}
		\caption{Comparison of the upper bounds for varying age limit $d$ in a \textit{single hop} for different service time distributions.}\label{fig:singleHop_vard}
	\end{figure*}
	
	\begin{figure*}[ht!]
		\centering
		\begin{subfigure}[b]{0.30\textwidth}
			\includegraphics[width=\textwidth]{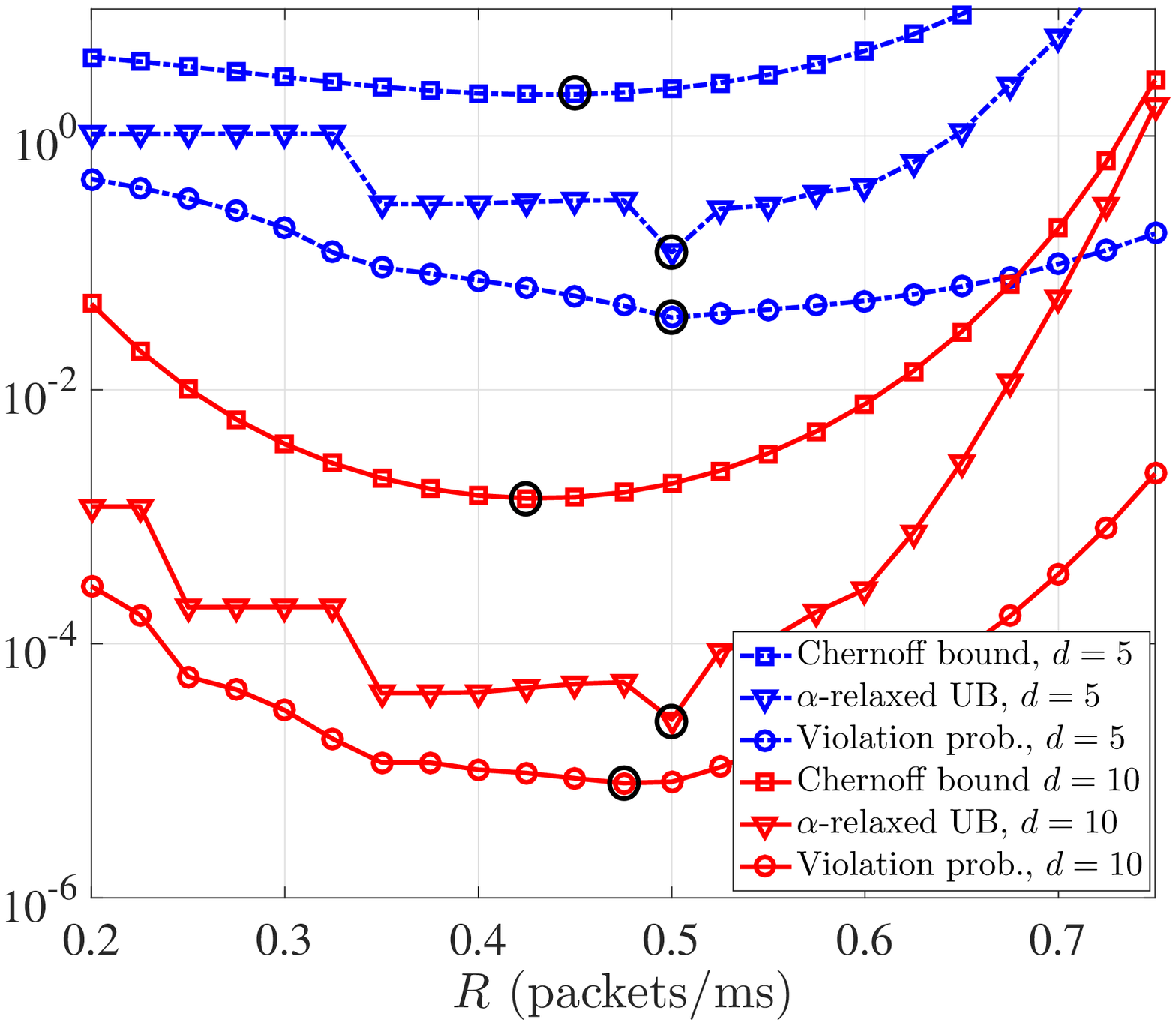}
			\caption{Geometric service with $p_1 = 0.85$ and  $p_2 = 0.9$.}
			\vspace{-0.0cm}
			\label{fig:twoHop_Geometric_varR}
		\end{subfigure} \quad 
		\begin{subfigure}[b]{0.30\textwidth}
			\includegraphics[width=\textwidth]{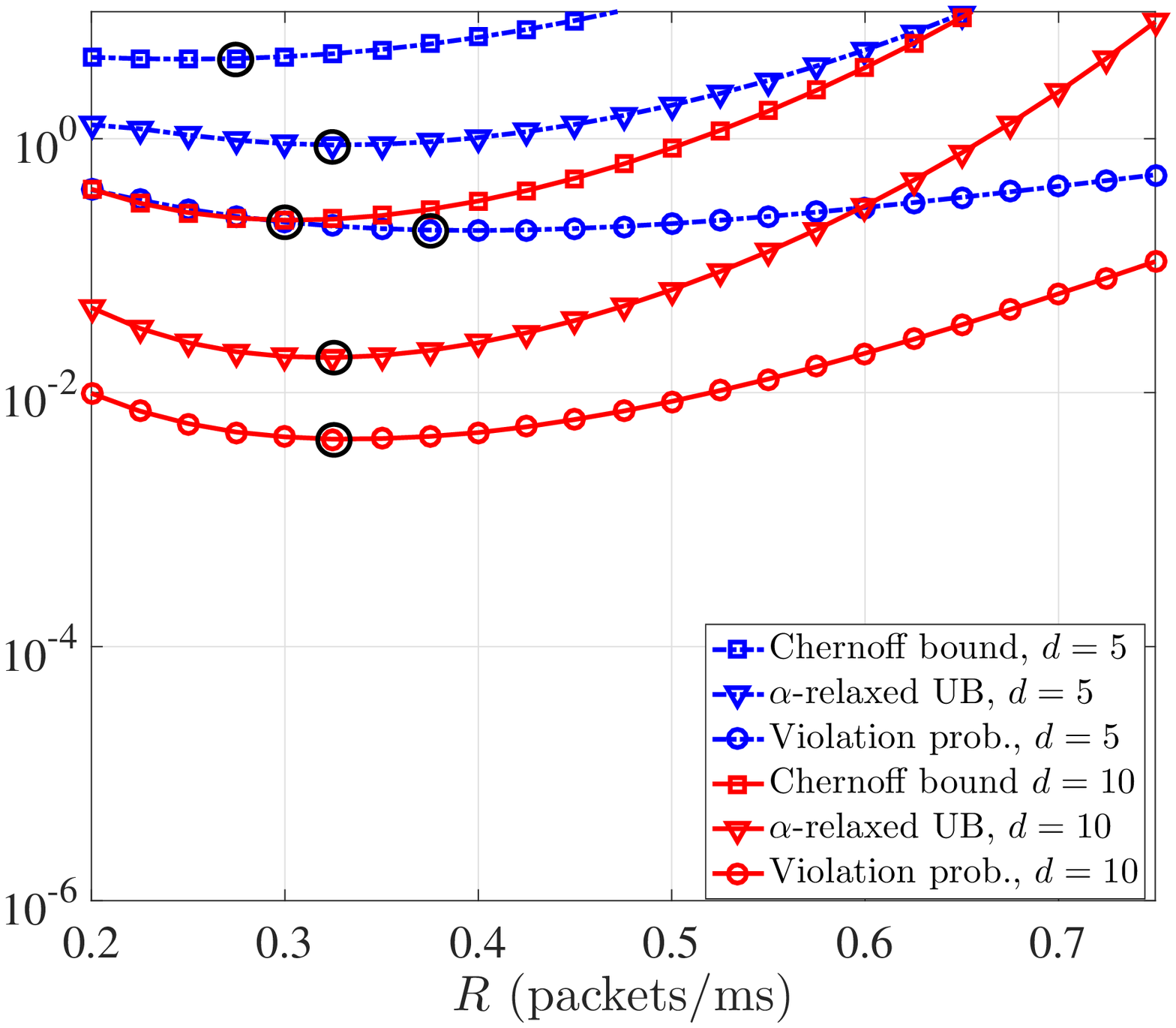}
			\caption{Exponential service with $\mu_1 = \mu_2 = 1$.}
			\vspace{-0.0cm}
			\label{fig:twoHop_Exp_varR}
		\end{subfigure} \quad 
		\begin{subfigure}[b]{0.30\textwidth}
			\includegraphics[width=\textwidth]{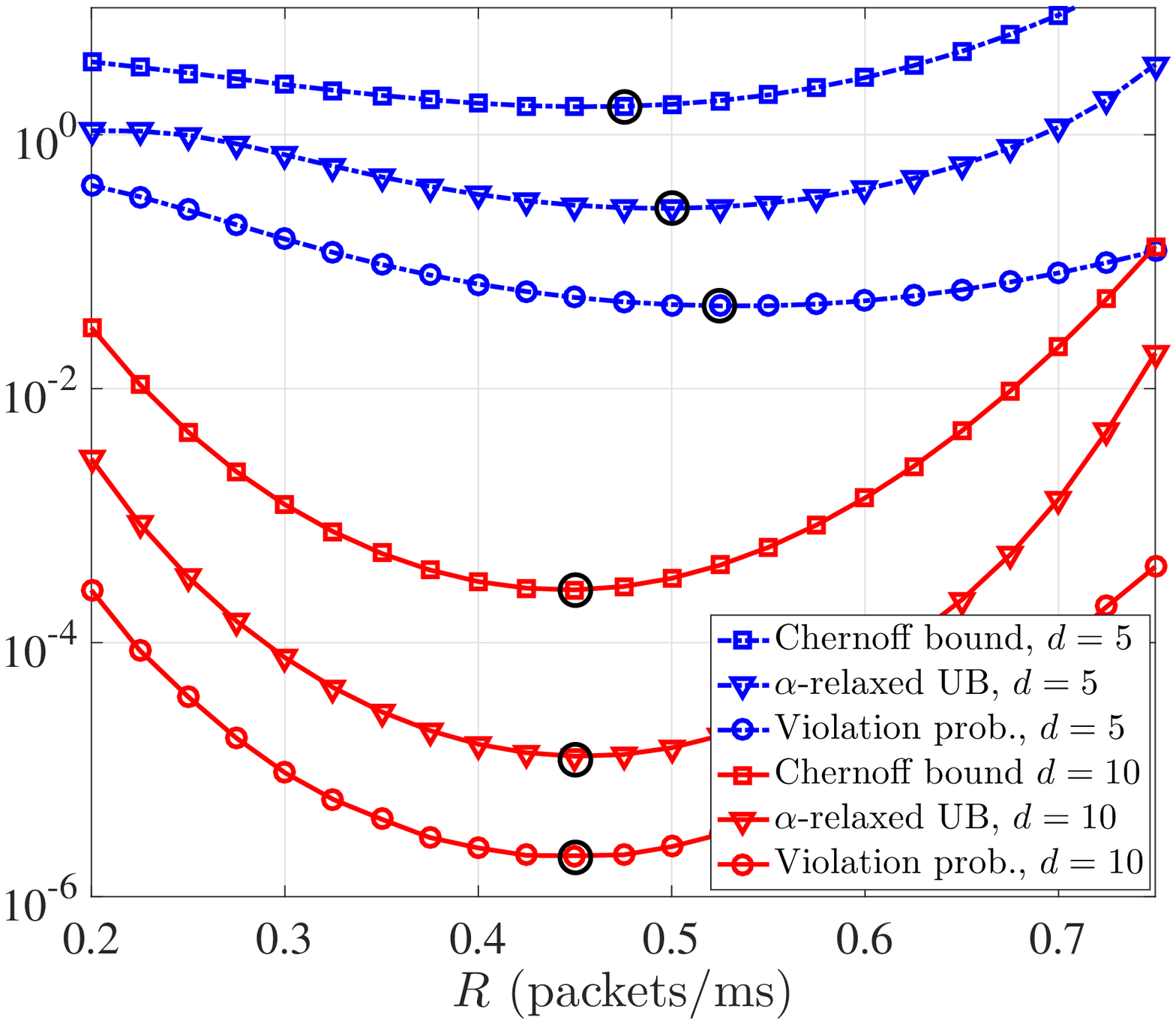}
			\caption{Erlang service with $b_1 = b_2 = 3$, $\lambda_1 = \lambda_2 = 3$, and $\mu_1 \! = \! \mu_2 =\! 1$.}
			\vspace{-0.0cm}
			\label{fig:twoHop_Erlang3_varR}
		\end{subfigure}
		\caption{Comparison of the upper bounds for varying arrival rate $R$ in a \textit{two hop} network for different service time distributions.}\label{fig:twoHop_varR}
	\end{figure*}
	
	\begin{figure*}[ht!]
		\centering
		\begin{subfigure}[b]{0.30\textwidth}
			\includegraphics[width=\textwidth]{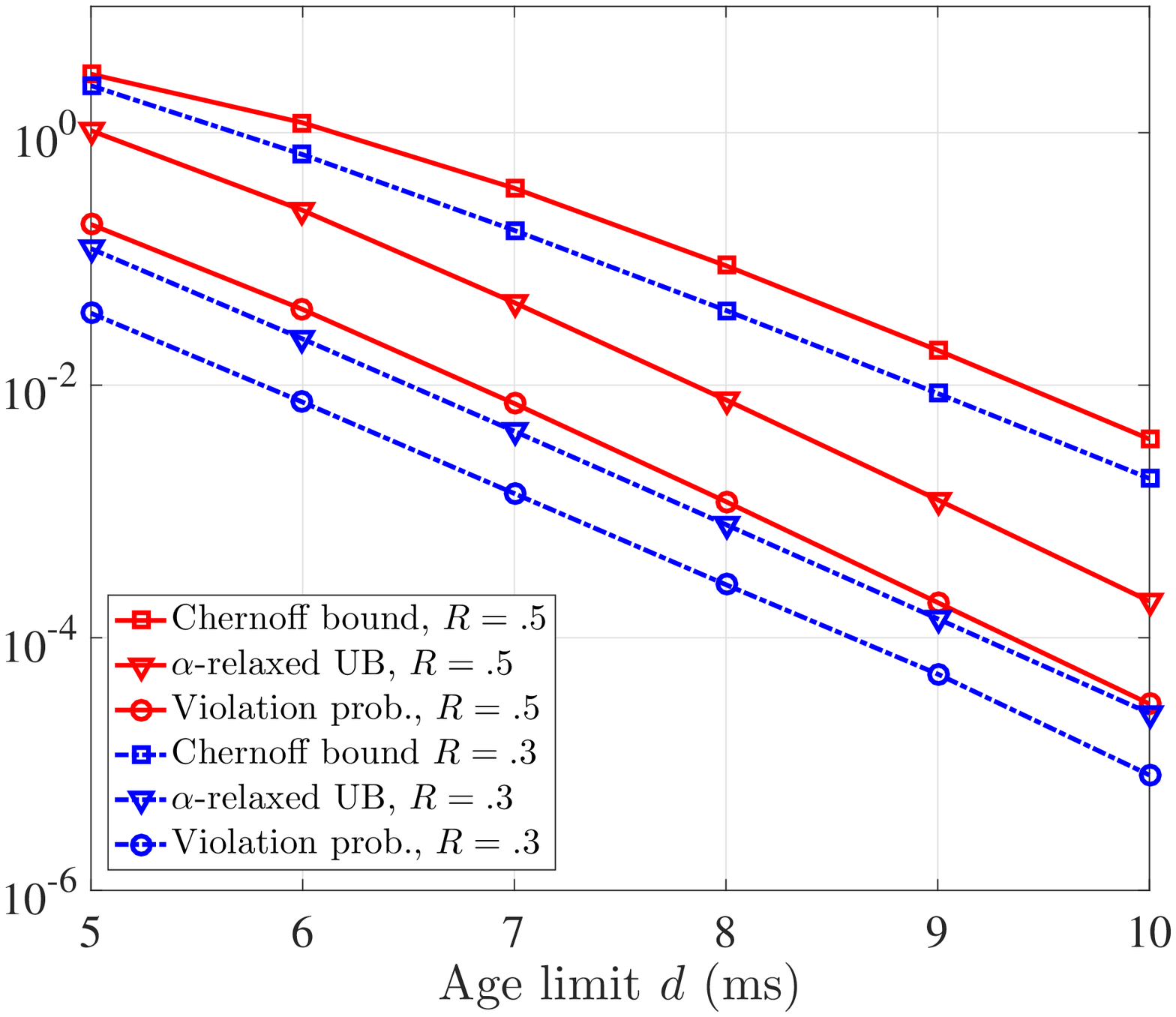}
			\caption{Geometric service with $p_1 = 0.85$ and  $p_2 = 0.9$}
			\vspace{-0.0cm}
			\label{fig:twoHop_Geometric_vard} 
		\end{subfigure} \quad 
		\begin{subfigure}[b]{0.30\textwidth}
			\includegraphics[width=\textwidth]{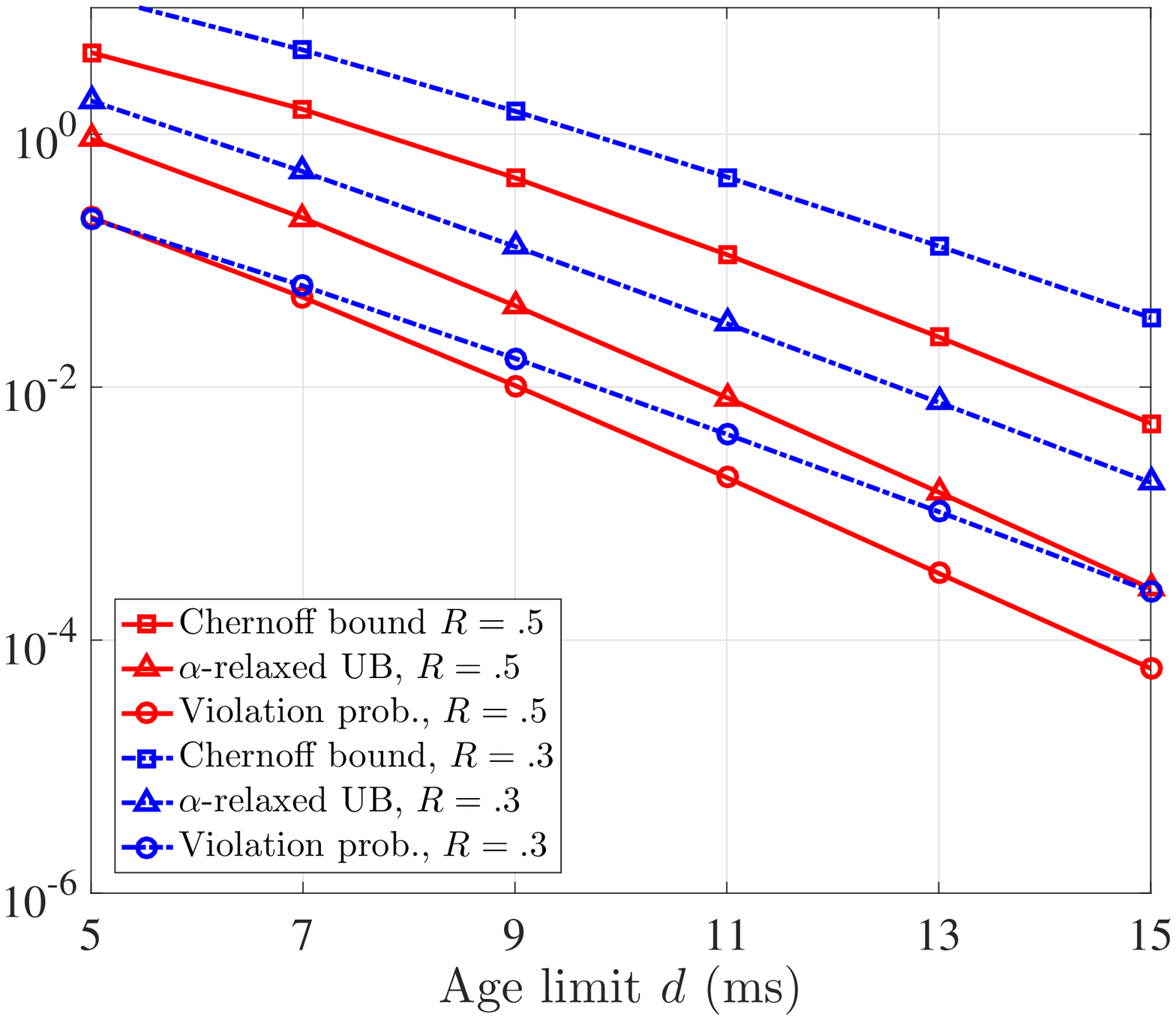}
			\caption{Exponential service with $\mu_1 = \mu_2 = 1$.}
			\vspace{-0.0cm}
			\label{fig:twoHop_Exp_vard}
		\end{subfigure} \quad 
		\begin{subfigure}[b]{0.30\textwidth}
			\includegraphics[width=\textwidth]{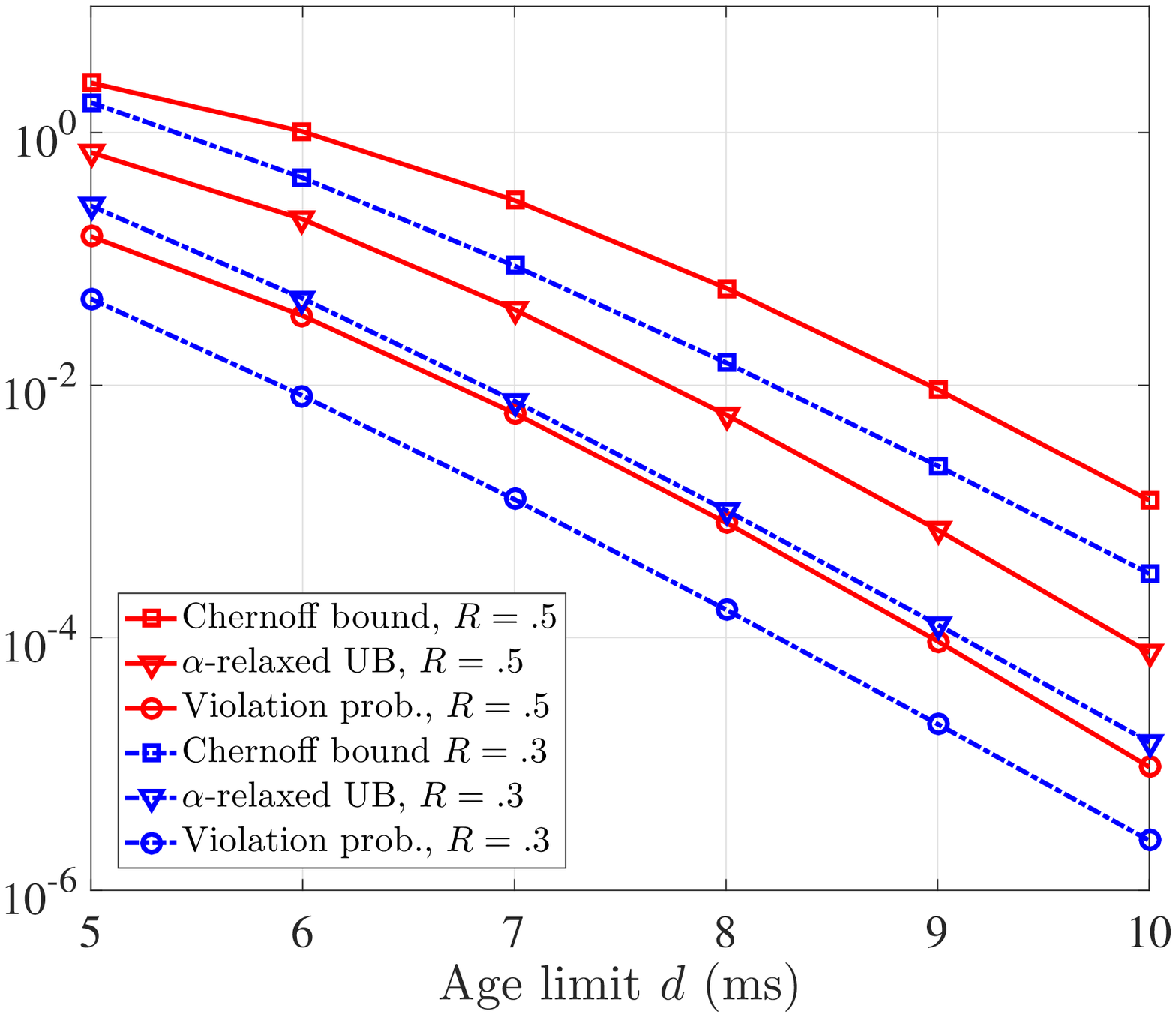}
			\caption{Erlang service with $b_1 = b_2 = 3$, $\lambda_1 = \lambda_2 = 3$, and $\mu_1 \! = \! \mu_2 =\! 1$.}
			\vspace{-0.0cm}
			\label{fig:twoHop_Erlang3_vard}
		\end{subfigure}
		\caption{Comparison of the upper bounds for varying age limit $d$ in a \textit{two hop} network for different service time distributions.}\label{fig:twoHop_vard}
	\end{figure*}
	
	\begin{figure*}[ht!]
		\centering
		\begin{subfigure}[b]{0.30\textwidth}
			\includegraphics[width = \textwidth]{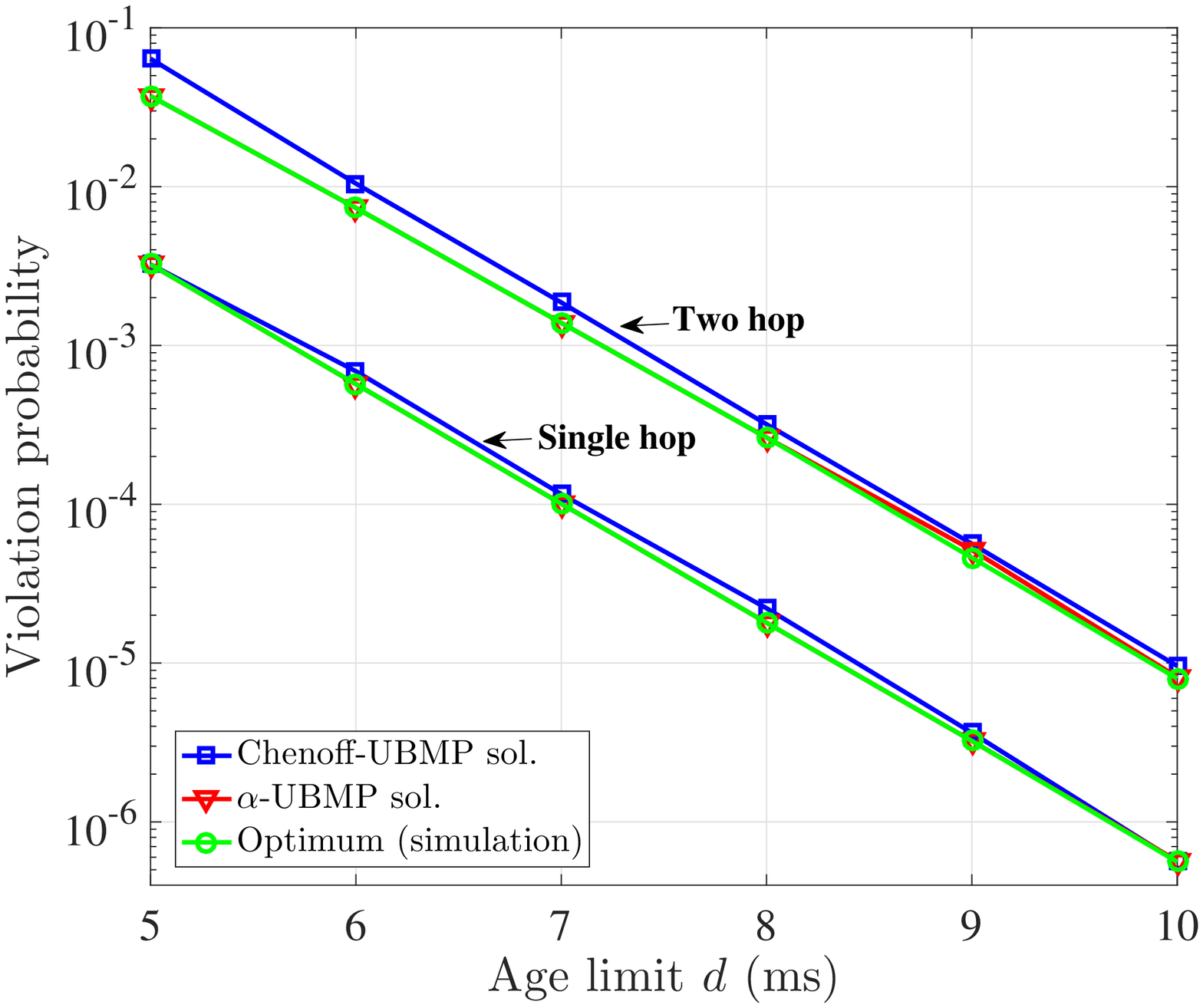}
			\caption{Geometric service with $p_1 = 0.85$ and  $p_2 = 0.9$.}
			\vspace{-0.0cm}
			\label{fig:comparisonSol_Geometric}
		\end{subfigure} \quad 
		\begin{subfigure}[b]{0.30\textwidth}
			\includegraphics[width = \textwidth]{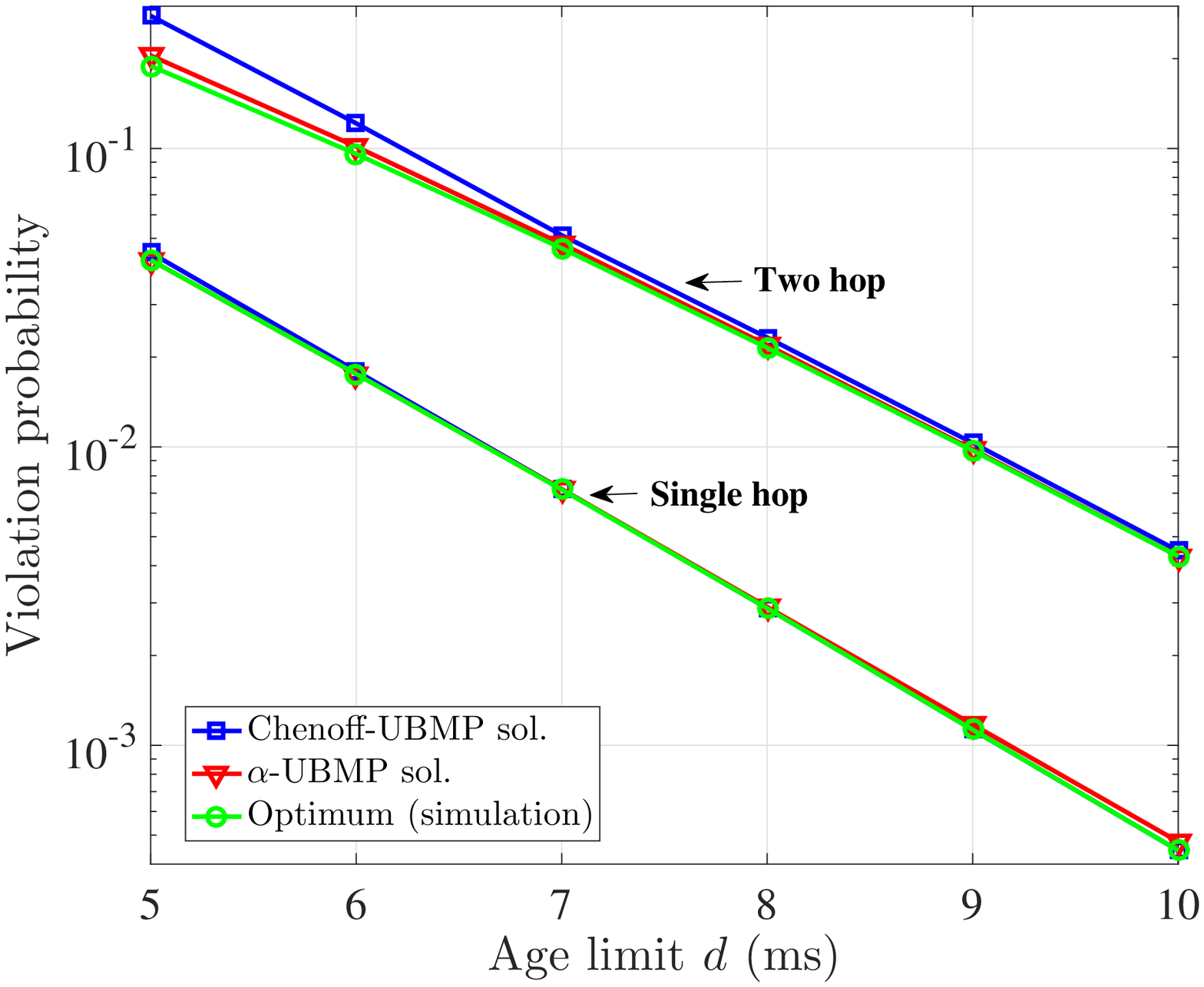}
			\caption{Exponential service with $\mu_1 = \mu_2 = 1$ packets/ms.}
			\vspace{-0.0cm}
			\label{fig:comparisonSol_exp}
		\end{subfigure} \quad 
		\begin{subfigure}[b]{0.30\textwidth}
			\includegraphics[width = \textwidth]{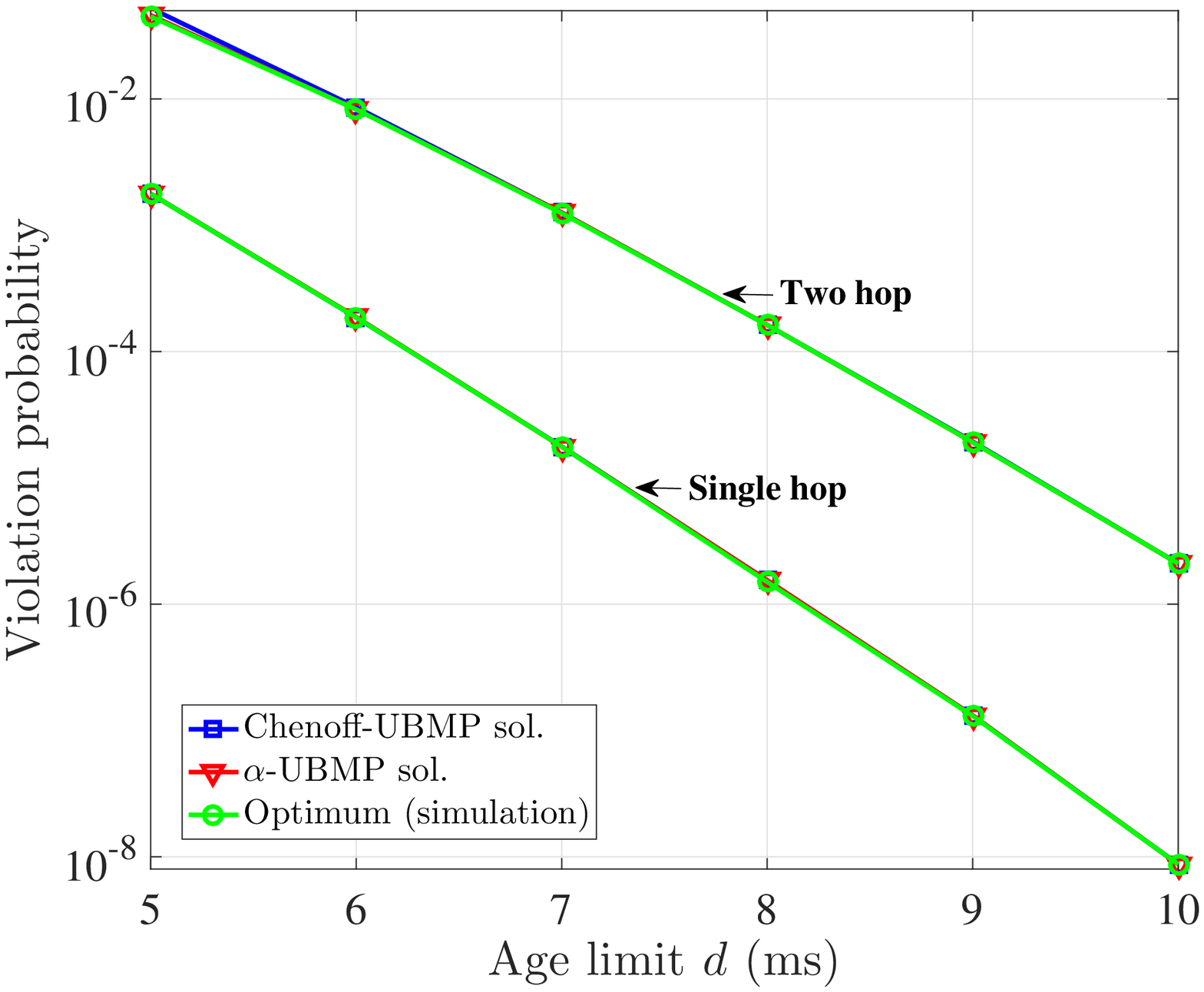}
			\caption{Erlang service with $b_1 = b_2 = 3$, $\lambda_1 = \lambda_2 = 3$, and $\mu_1 \! = \! \mu_2 =\! 1$.}
			\vspace{-0.0cm}
			\label{fig:comparisonSol_Erlang3}
		\end{subfigure}
		\caption{Evaluation of the rate solutions obtained using upper bound minimization for different service time distributions.}\label{fig:comparisonSol}
	\end{figure*}
	
	\subsection{Properties of Upper Bounds}
	\subsubsection{Single Hop}
	In Figures~\ref{fig:singleHop_varR} and~\ref{fig:singleHop_vard}, we present the upper bounds and the simulated AoI violation probability for varying arrival rate $R$ and varying age limit $d$ for different distributions for the single-hop scenario. From Figure~\ref{fig:singleHop_varR}, we observe that the upper bounds and the violation probability have convex nature and a global minimum \jpcolor{(highlighted in black circles)} in the chosen range of $R$. Further, observe that the curvature of the upper bounds approximately follow the curvature of the simulated violation probability around its minimum value and only deviates at higher sampling rate. This is an interesting property as it suggests that a rate that minimizes the upper bound \jpcolor{will be a ``good" rate solution for minimizing} the violation probability. 
	We note that the $\alpha$-relaxed upper bound curves are not continuous because the probability terms $\Phi(v_0,v_1,R)$ involves a floor function, namely, $\lfloor \beta \rfloor$.
	From Figure~\ref{fig:singleHop_vard}, we observe that the decay rates of the upper bounds match closely the decay rate of the violation probability. This further strengthens our statement above that minimizing the upper bounds results in good heuristic rate solutions for the considered range of age limits.
	
	\subsubsection{Two Hop}
	In Figures~\ref{fig:twoHop_varR} and~\ref{fig:twoHop_vard}, we present the upper bounds and the simulated AoI violation probability for varying arrival rate $R$ and varying age limit $d$ for different distributions for the two-hop scenario. We observe similar trends as in the case of single-hop scenario. Nevertheless, the bounds become relatively looser. This can be attributed to the fact that the union bound is applied twice for the two-hop scenario.
	
	Note that for both single-hop and two-hop scenarios $\alpha$-relaxed bound is much lower than the Chernoff bound. Nevertheless, Chernoff bound can be useful for the cases where the exact distribution of the summation of service times is intractable. 
	
	\jpcol{\subsubsection{Service Times with Higher Variance}
		In this section, we study how the upper bounds perform for service time with higher variance. In Figure~\ref{fig:twoHop_Exp_varR_heterogenous}, we consider heterogeneous exponential service times, with $\mu_1 = 0.75$ and $\mu_2 = 1$. We have chosen  $\mu_1 = 0.75$ so that the variance in this case is higher than that of the homogeneous case where both  $\mu_1$ and $\mu_2$ are equal to $1$. 
		We note that the trend persists and that the mismatch in the minima among the three curves in this case becomes greater when lower age limit is considered, i.e., $d=5$ ms. However, the mismatch is much smaller at higher $d$. We also note that compared to the homogeneous server case in Figure~\ref{fig:twoHop_Exp_varR}, heterogeneity does not affect the conclusions regarding system behaviour with respect to AoI. Nevertheless, the performance becomes more dependent on the bottleneck link in this case. 
		
		In Figure~\ref{fig:variance}, we consider hyper-exponential service time distribution with probability density function given by $p\lambda_1e^{-\lambda_1 x} + (1-p)\lambda_2e^{-\lambda_2 x}$. We choose $p = 0.91$, $\lambda_1 = 0.95$, and $\lambda_2 = 2$ such that the mean value is equal to $1$ ms. We note that this distribution has higher variance compared to exponential-service time distribution with mean $1$. For computing the alpha-relaxed upper bound, we set $K=6$ and numerically evaluated the convolution of hyper-exponential  probability distribution functions to obtain values for $\phi(v_0,v_1,R)$.
		
		From both Figures~\ref{fig:twoHop_Exp_varR_heterogenous} and~\ref{fig:variance}, we observe that for the two-hop scenario, for $d = 5$, the mismatch between the heuristic rate solution provided by $\alpha$-UBMP and the optimal rate solution is relatively bigger. Nevertheless,  under these settings, it should be noted that the value of the minimum AoI violation probability is not significantly lower than that  achieved by the heuristic rate solution. Again, the main trends noticed with the other three service distributions, see Figure~\ref{fig:twoHop_varR}, are prevailing here as well. 
		
	}
	
	\begin{figure}
		\centering
		\includegraphics[width = 3in]{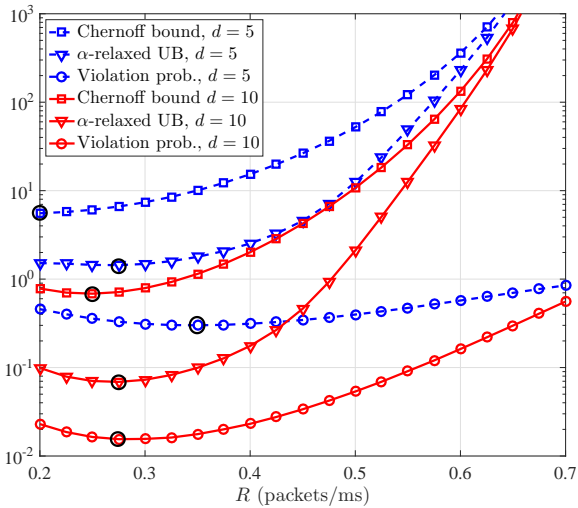}
		\caption{Comparison of the upper bounds for varying arrival rate $R$ in a \textit{two hop} network for heterogeneous exponential service time distributions, with $\mu_1 = 0.75$ packet/ms and $\mu_2=1$ packets/ms.}
		\label{fig:twoHop_Exp_varR_heterogenous}
	\end{figure}

	\begin{figure*}[ht!]
	\centering
	\begin{subfigure}[b]{0.4\textwidth}
		\includegraphics[width=\textwidth]{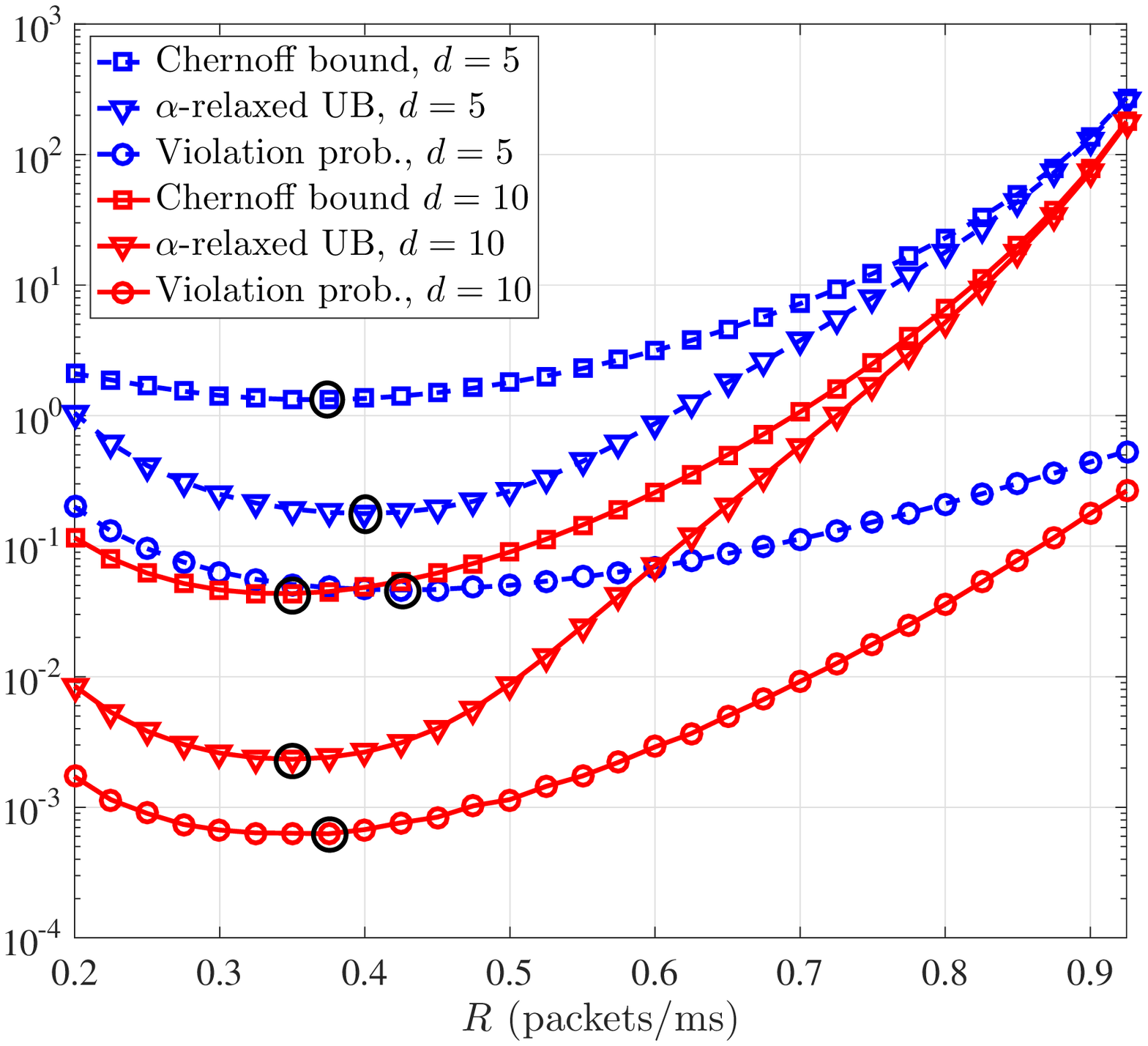}
		\caption{Single hop}
		\vspace{-0.0cm}
		\label{fig:singleHop_HyperExp_varR}
	\end{subfigure} \quad \quad \quad \quad
	\begin{subfigure}[b]{0.4\textwidth}
		\includegraphics[width=\textwidth]{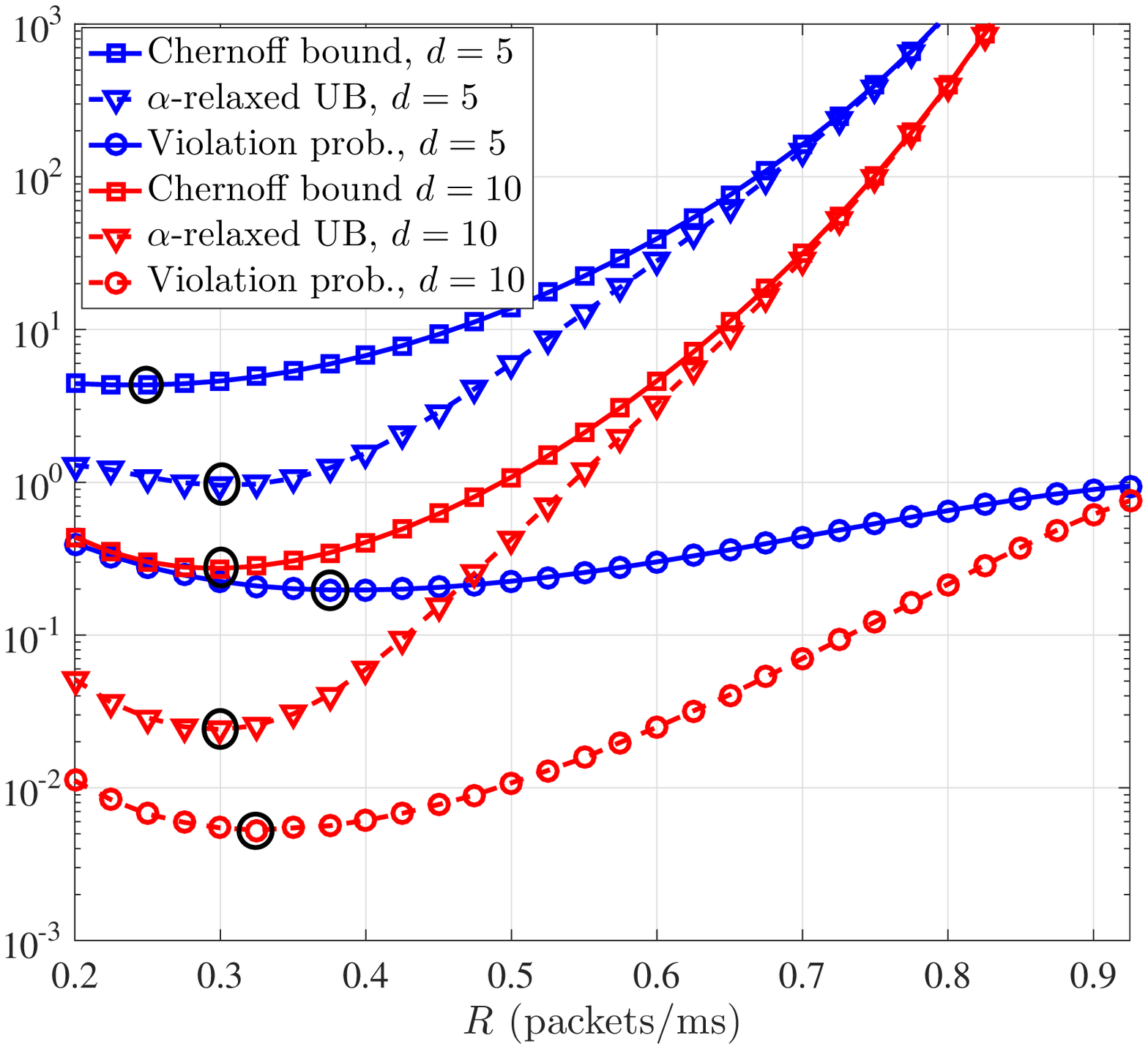}
		\caption{Two hop}
		\vspace{-0.0cm}
		\label{fig:twoHop_HyperExp_varR}
	\end{subfigure}
	\caption{Comparison of the upper bounds for varying arrival rate $R$ for hyper-exponential service-time distribution.}\label{fig:variance}
	\end{figure*}
	
	\subsection{Quality of the Heuristic Solution}
	In Figure~\ref{fig:comparisonSol}, we compare the violation probabilities for rate solutions obtained by solving the UBMPs and the estimated minimum/optimum violation probability obtained by exhaustive search using simulation, for both single-hop and two-hop scenarios. Note that the difference between the violation probabilities achieved by the heuristic rate solutions  and the optimum violation probability is negligible. This suggests that the solutions of the UBMPs are near optimal for $\mathcal{P}$ \jpcolor{for the considered service-time distributions}.  
	This can be attributed to the fact that the upper bounds have decay rate that matches the decay rate of the violation probability as stated before. 
	Although $\alpha$-relaxed upper bound is much lower than Chenoff bound the solutions of $\alpha$-UBMP provide only slightly lower violation probability than that of the Chernoff-UBMP solutions. Thus, Chernoff-UBMP is relatively tractable and the rate solutions provided can be used as first step toward computing close-to-optimal solutions by utilizing additional information about the service distributions. 
	
	\textbf{\textit{Remark 3:}} We note that unlike the time-average age objective, which is minimized at $0.515$ utilization factor ($\lambda_1/\mu_1$) for the D/M/1 queue~\cite{kaul_2012b}, the optimal rate solution and in turn the utilization factor that minimizes AoI violation probability depends on age limit $d$. For a comparison, in Figure~\ref{fig:comparisonSol_exp} the single-hop scenario is equivalent to D/M/1 system and in this case the optimal utilization factors are $\{0.425,0.4,0.4,0.35,0.35,0.35,0.35,0.35,0.35\}$.
	
	\begin{figure*}[ht!]
	\centering
	\begin{subfigure}[b]{0.42\textwidth}
		\includegraphics[width=\textwidth]{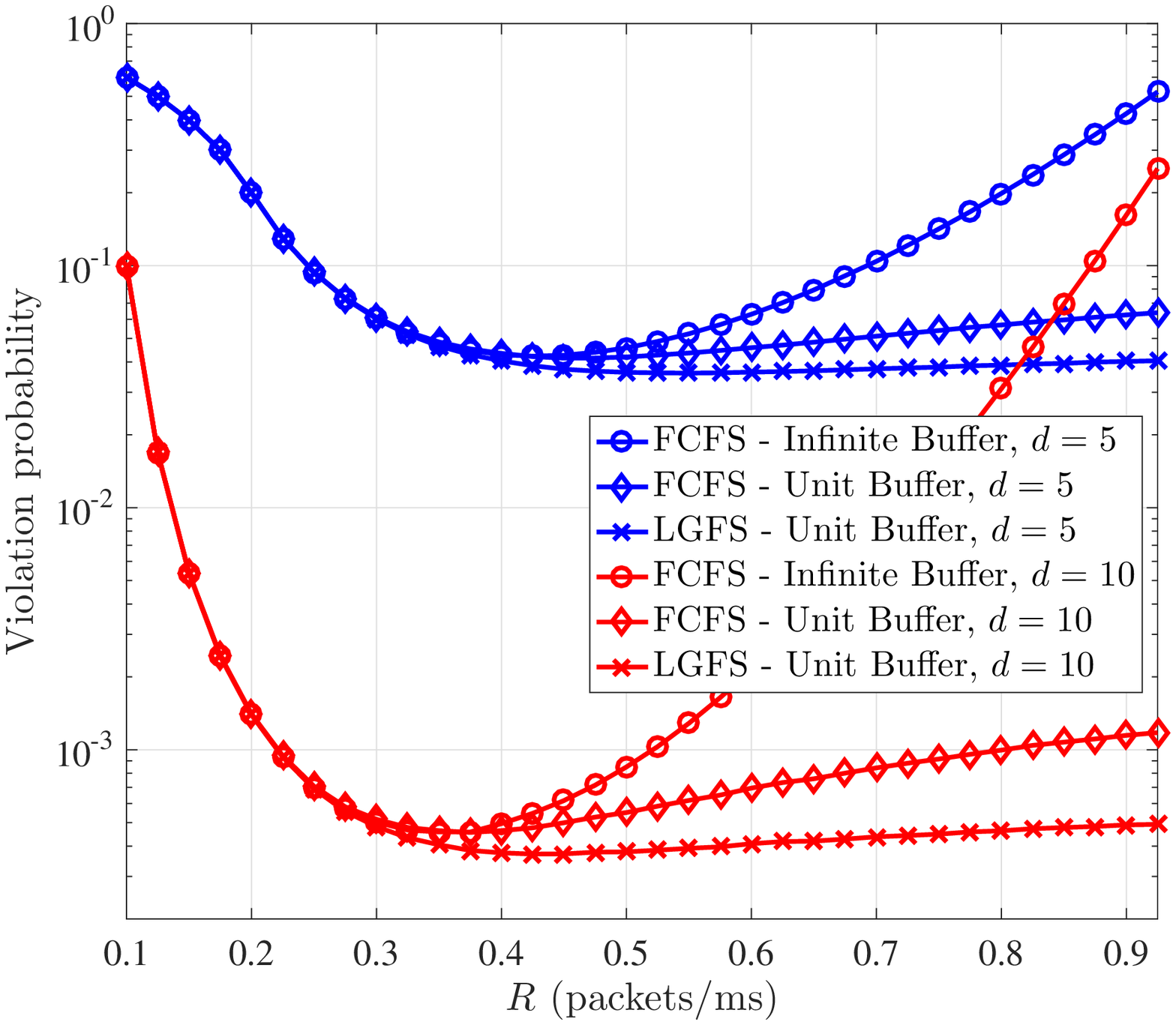}
		\caption{Single hop.}
		\vspace{-0.0cm}
		\label{fig:singleHop_unitBuffer}
	\end{subfigure} \quad \quad \quad 
	\begin{subfigure}[b]{0.42\textwidth}
		\includegraphics[width=\textwidth]{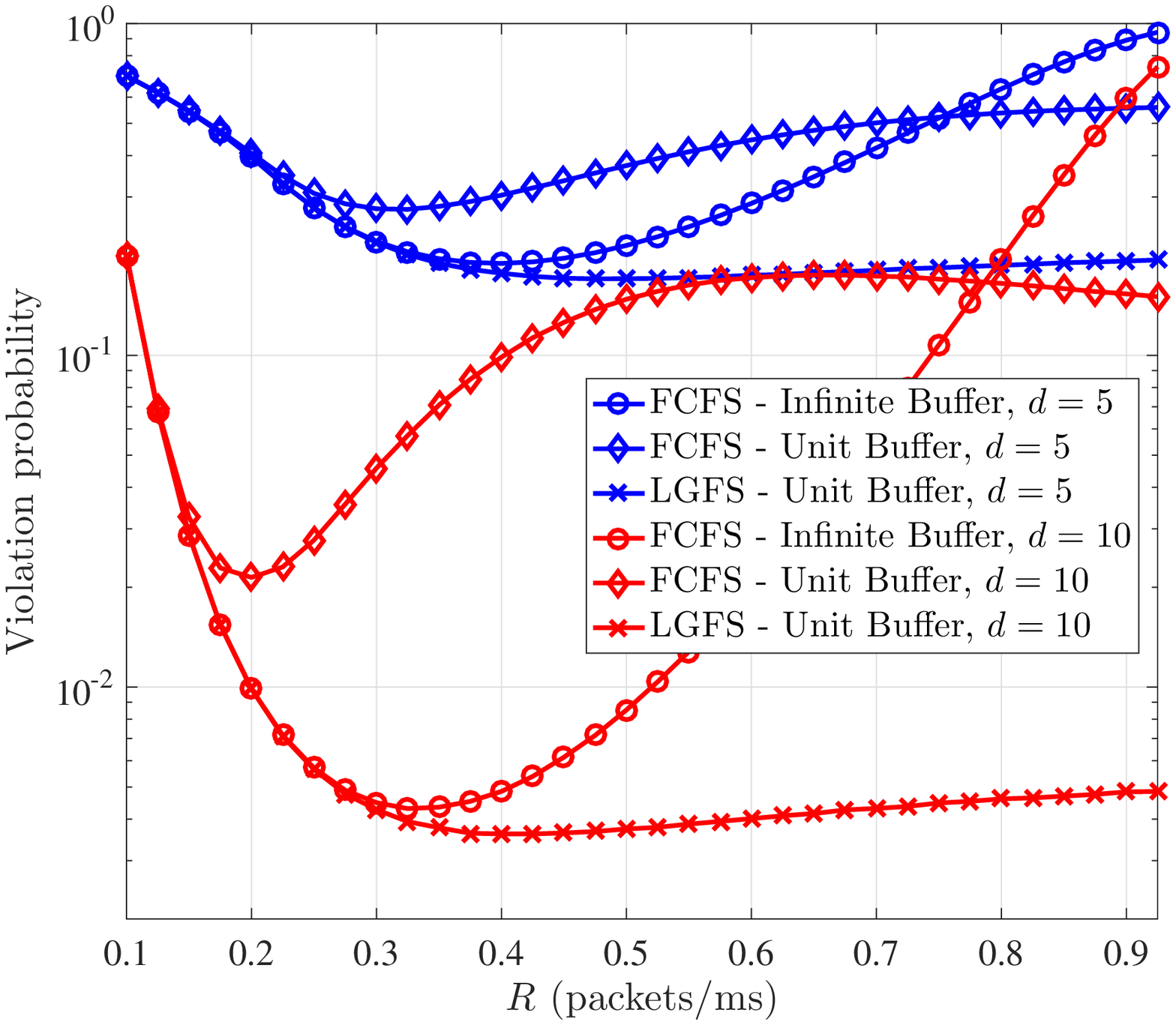}
		\caption{Two hop.}
		\vspace{-0.0cm}
		\label{fig:twoHop_unitBuffer}
	\end{subfigure}
	\caption{Comparison of AoI violation probability achieved under different queue management policies under exponential-service times with rates equal to one packet/sec.}\label{fig:unitBuffer}
	\end{figure*}	
	
	\jpcol{\subsection{Queue Management Policies with Unit Buffer}
		Although this work is dedicated to study AoI distribution for WNCS  under FCFS\footnote{By FCFS we mean  First-Come-First-Serve with infinite buffer.} scheduling, 
		recent research results have shown that considering a unit buffer with queue management policies provide lower AoI statistics in comparison to FCFS with infinite buffer~\cite{Costa_2016,Champati_GG1_2019,Bedewy_2017a,Bedewy2019}. In this section, we investigate this effect by applying two widely referenced (in the context of AoI research) queue management policies, namely,  FCFS-Unit Buffer  and LGFS-Unit Buffer, to every queue in our network and then compare the achieved  AoI violation probability to that we obtained earlier under FCFS. 
		Both of these policies employs a one-packet buffer, however, they differ in that,  whenever the buffer is occupied and a new  packet arrives, the existing packet is kept in the first while it is replaced with the newly arriving packet in the second. 
		In~\cite{Bedewy_2017a}, it was shown that LGFS-Unit Buffer minimizes AoI processes, in stochastic ordering sense, among all non-preemptive service policies for any arrival process and service-time distributions. This implies, for the tandem two-queue system we consider, LGFS-Unit Buffer results in minimum AoI violation probability. Hence, it provides a good reference to measure the performance of other queue management policies against.
		
		In Figures~\ref{fig:singleHop_unitBuffer} and~\ref{fig:twoHop_unitBuffer}, 
		the AoI violation probability  is plotted against the arrival rate $R$ for the FCFS as well as the two unit buffer queue management policies mentioned above, assuming  exponential-service times with rate $\mu = 1$ packet/sec and for two age limits $d=\{5, 10\}$ ms. 
		We observe that, the minimum AoI violation probability under FCFS-Unit Buffer and LGFS-Unit Buffer is comparable to that under FCFS in the single-hop scenario. However, the performance of FCFS-Unit Buffer deteriorates drastically  in the two-hop case compared to the other two, see Figure~\ref{fig:twoHop_unitBuffer}. This can be attributed to the fact that, under FCFS-Unit Buffer, packets that are served at first link may still be dropped when arriving at the second link if its buffer is already occupied. This effect may be exacerbated when more links (hops) are added to the tandem and even more when heterogeneous service processes are present at different links along the cascade where latter links experience higher utilization. This example shows that in tandem-queuing systems, it is not always true that FCFS with finite buffer has lower AoI statistics than that of FCFS with infinite buffer. We contrast this with the performance trends of these policies for a system with parallel servers between the source and the destination~\cite{Bedewy2019}, where it was demonstrated that, under Poisson arrivals and exponential service-times, the minimum values for average AoI and average peak AoI achieved under FCFS with infinite buffer are much higher than that for the case with finite buffer.
		
		Furthermore, in both  scenarios, we observe that the  AoI violation probability under FCFS with infinite queue is quite close to that of LGFS-Unit Buffer at low utilization (i.e., low $R$), and more importantly the minima for both  cases are reasonably close and are achieved around the same arrival rate $R$. The reason for such behaviour can be attributed to the fact that at such low arrival rate the buffer would be empty most of the time mimicking the unit buffer behaviour. In the two-hop scenario, the FCFS would still have  buffer space at the second hop to hold (and not drop as FCFS-Unit Buffer may do) a packet that is successfully forwarded by the first hop, hence, countering the performance deterioration experienced by FCFS-Unit Buffer due to packet drops that we highlighted in the two-hop scenario above. 

		The above observations are quite interesting. They suggest, at least for deterministic arrivals and tandem-queuing system, using FCFS with infinite queue may achieve a minimum AoI violation probability that is reasonably close to the optimum (achieved by LGFS-Unit Buffer among the set of non-preeemptive policies~\cite{Bedewy_2017a}). This also opens an interesting research question for future work: how far  the minimum AoI violation probability achieved under FCFS can be from the optimum?
		

	}

	\section{Conclusion and Future Work}\label{sec:conclusion}
	We provide a general characterization of AoI violation probability for a network with periodic input arrivals. Using this characterization, we formulate an optimization problem $\mathcal P$ to find the optimal input rate which minimizes the AoI violation probability. Further, we show that $\mathcal{P}$ is equivalent to the problem of minimizing the violation probability of the departure time of a tagged arrival $\nR$ over the rate region $[\frac{1}{d},\mu)$. 
	Noting that computing an exact expression for the violation probability is hard, 
	we propose an Upper Bound Minimization Problem (UBMP) and its more computationally tractable versions Chernoff-UBMP and \prob, which result in heuristic rate solutions. We also present the Chernoff-UBMP for N-hop tandem queuing system.
	We solve Chernoff-UBMP and \prob~for single-hop and two-hop scenarios for three service-time distributions, namely, geometric, exponential and Erlang. Numerical results suggest that the rate solutions of \prob~are near optimal for $\mathcal{P}$, demonstrating the efficacy of our method. 
	
	{\color{black} Furthermore, our simulation results suggest that, FCFS performs close to the optimum achieved by LGFS-Unit Buffer and drastically outperforms FCFS-Unit Buffer in a two-hop network with respect to AoI violation probability. 
		This opens up an interesting research question: how far will be the minimum AoI violation probability achieved under FCFS from the optimum?
	}
	Another interesting research direction for future work would be to extend our results to stochastic arrivals. 
	We are also studying the computational complexity for solving \prob~and investigating more efficient solution methods, i.e., by identifying the range of $\alpha$ for which a good heuristic solution for $\mathcal{P}$ can be obtained. 
	{\color{black} Finally, we would like to investigate different queuing disciplines relevant to AoI.}
	

	\appendix
{\allowdisplaybreaks	
	\subsection{Proof of Lemma~\ref{lem:twoHop}}\label{lem:twoHop:proof}
	The violation probability is given by
	\begin{align}\label{eq1:twohop}
	&\P\{D(\nR) > t\} = \P\{D_{2}(\nR) > t\}  \nonumber\\
	&= \P \left\{\max_{0 \leq v_0 \leq \nR}  \left (A_{2}(\nR - v_0) + \sum_{i=0}^{v_0}X_{2}^{\nR - i}\right)  >  t \right\}\nonumber \\
	&= \P \left\{\bigcup\limits_{v_0=0}^{\nR}  \left (A_{2}(\nR - v_0) + \sum_{i=0}^{v_0}X_{2}^{\nR - i}  >  t \right)\right\}\nonumber \\
	&\leq \sum_{v_0=0}^{\nR} \P \left\{A_{2}(\nR - v_0) + \sum_{i=0}^{v_0}X_{2}^{\nR - i}  >  t \right\}.
	\end{align}
	Further, we have 
	\begin{align*}
	&A_{2}(\nR - v_0) = D_{1}(\nR - v_0) \\
	&= \max_{0 \leq v_1 \leq \nR - v_0} \left (A_{1}(\nR - v_0 - v_1) + \sum_{i=0}^{v_1}X_{1}^{\nR - v_0 - i}\right).
	\end{align*}
	Substituting $A_{2}(\nR - v_0)$ and $A_{1}(\nR - v_0 - v_1) = \frac{\nR - v_0 - v_1}{R}$ in~\eqref{eq1:twohop} we obtain
	{\allowdisplaybreaks\begin{align*}
		&\P\{D(\nR) > t\}\nonumber \\
		& \leq \sum_{v_0=0}^{\nR}\!\! \P \Big[\max_{0 \leq  v_1 \leq  \nR - v_0}\! \Big ( \frac{\nR - v_0 - v_1}{R}\! + \! \sum_{i=0}^{v_1} X_{1}^{\nR - v_0 - i}\Big)  \nonumber \\
		& \quad \quad \quad \quad \quad \quad \quad \quad  + \sum_{i=0}^{v_0} X^{N}_{\nR - i}\!  > \! t \Big] \nonumber\\
		&\leq \sum_{v_0=0}^{\nR}\sum_{v_1=0}^{\nR - v_0} \P \Big[\frac{\nR - v_0 - v_1}{R} + \sum_{i=0}^{v_1}X_{1}^{\nR - v_0 - i} \nonumber \\
		& \quad \quad \quad \quad \quad \quad \quad \quad + \sum_{i=0}^{v_0}X_{2}^{\nR - i}  >  t \Big] \nonumber\\
		&= \sum_{v_0=0}^{\nR}\sum_{v_1=0}^{\nR - v_0} \P \left\{\sum_{i=0}^{v_1}X_{1}^{i} + \!\! \sum_{i=0}^{v_0}X_{2}^{i}  >  d \!+\! \frac{v_0 \! + \! v_1 \! - \! 1}{R} \right\}.
		\end{align*}}
	In the second step above, we have used the union bound. In the last step we have used $\nR \leq R(t-d) + 1$. Also, since $X_1$ and $X_2$ are i.i.d., we re-indexed the superscripts of $X_1$ and $X_2$ in the summations. The result follows from the fact that as $t$ goes to infinity $\nR$ goes to infinity.

\subsection{Proof of Theorem~\ref{thm:twoHop:Chernoff}}\label{thm:twoHop:Chernoff:proof}	
We first obtain Chernoff bound for $\Phi(v_0,v_1,R)$. We have
	\begin{align}\label{eq5:twohop}
	&\Phi(v_0,v_1,R) = \P \left\{\sum_{i=0}^{v_1}X_{1}^{i} + \!\! \sum_{i=0}^{v_0}X_{2}^{i}  >  d \!+\! \frac{v_0 \! + \! v_1 \! - \! 1}{R} \right\} \nonumber \\
	&\leq \min_{s>0}  e^{-s( d \!+\! \frac{v_0 \! + \! v_1 \! - \! 1}{R})} \mathbb{E}[e^{s(\sum_{i=0}^{v_1}X_{1}^{i} + \! \sum_{i=0}^{v_0}X_{2}^{i})}]\nonumber\\
	&= \min_{s>0}  e^{-s(d \!+\! \frac{v_0 \! + \! v_1 \! - \! 1}{R})}[M_{1}(s)]^{v_1 + 1}[M_{2}(s)]^{v_0 + 1} \nonumber\\
	&=  \min_{s>0} e^{-s(d \! - \! \frac{1}{R})}M_{1}(s)  M_{2}(s)\beta^{v_1}_1(s)\beta^{v_0}_2(s).
	\end{align}
	Assuming the moment generating function of $X$ exists, in the second step above we have used the Chernoff bound. In the third step above we have used the fact that $X^{k}_{n}$ are i.i.d. for all $k$ and $n$, and in the last step we have used~\eqref{eq:beta}. Using~\eqref{eq5:twohop} in Lemma~\ref{lem:twoHop}, we obtain
	\begin{align}\label{eq6:twohop}
	&\lim_{t \rightarrow \infty} \P\{D(\nR) > t\} \leq  \lim_{\nR \rightarrow \infty} \sum_{v_0=0}^{\nR}\sum_{v_1=0}^{\nR - v_0} \Phi(v_0,v_1,R) \nonumber \\
	&\leq \min_{s>0} e^{-s(d \! - \! \frac{1}{R})}M_{1}(s)  M_{2}(s) \phi(s,\beta_1(s),\beta_2(s)), 
	\end{align}
	where
	\begin{align}\label{eq7:twohop}
	\phi(s,\beta_1(s),\beta_2(s)) = \lim_{\nR \rightarrow \infty} \! \sum_{v_0=0}^{\nR}\sum_{v_1=0}^{\nR - v_0}\!\! \beta^{v_1}_1(s)\beta^{v_0}_2(s).
	\end{align}
	Note that in the second step of~\eqref{eq6:twohop} we have used the fact that for positive quantities sum over minimum is less than or equal to minimum over the sum. In the following lemma we provide a closed form expression for $\phi(s,\beta_1(s),\beta_2(s))$.
	\begin{lemma}\label{lem:infSum}
		For $s \in \mathcal{S}$, 
		\begin{align*}
		\phi(s,\beta_1(s),\beta_2(s))  = \frac{1}{(1 - \beta_1(s))(1 - \beta_2(s))}.
		\end{align*}
	\end{lemma}
	\begin{proof}
		Recall that $\beta_1(s) < 1$ and $\beta_1(s) < 1$, for all $s \in \mathcal{S}$. Using this, we obtain
		\begin{align*}
		&\phi(s,\beta^{v_1}_1,\beta^{v_0}_2) = \lim_{\nR \rightarrow \infty} \! \sum_{v_0=0}^{\nR}\sum_{v_1=0}^{\nR - v_0}\!\! \beta^{v_1}_1(s)\beta^{v_0}_2(s) \\
		&= \lim_{\nR \rightarrow \infty} \! \sum_{v_0=0}^{\nR} \beta^{v_0}_2(s) \sum_{v_1=0}^{\nR - v_0}\!\! \beta^{v_1}_1(s) \\
		&= \lim_{\nR \rightarrow \infty} \! \sum_{v_0=0}^{\nR} \beta^{v_0}_2 \cdot \frac{(1-\beta^{\nR - v_0 + 1}_1(s))}{1-\beta_1(s)} \\
		&= \lim_{\nR \rightarrow \infty} \! \sum_{v_0=0}^{\nR}\left [ \frac{\beta^{v_0}_2}{1-\beta_1(s)} -  \frac{\beta^{v_0}_2 (s) \beta^{\nR - v_0 + 1}_1(s)}{1-\beta_1(s)}\right] \\
		&=\! \frac{1}{(1 \! -\! \beta_1(s))(1\! -\! \beta_2(s))}\! -\! \lim_{\nR \rightarrow \infty} \! \sum_{v_0=0}^{\nR}\!\frac{\beta^{v_0}_2 \!(s) \beta^{\nR - v_0 + 1}_1\!(s)}{1-\beta_1(s)}
		\end{align*}
		It is now sufficient  to show that the summation term above is equal to zero. We first note that the summation is non-negative since $0 \leq \beta_1(s) < 1$ and $0 \leq \beta_1(s) < 1$. Let $\beta(s) = \min(\beta_1(s),\beta_1(s))$, then we have
		\begin{align*}
		&\lim_{\nR \rightarrow \infty} \! \sum_{v_0=0}^{\nR}\frac{\beta^{v_0}_2 (s) \beta^{\nR - v_0 + 1}_1(s)}{1-\beta_1(s)} \leq \lim_{\nR \rightarrow \infty} \! \sum_{v_0=0}^{\nR}\frac{\beta^{\nR + 1}(s)}{1-\beta_1(s)} \\
		&= \lim_{\nR \rightarrow \infty}  \frac{(\nR + 1)}{ \beta^{-(\nR + 1)}(s)(1-\beta_1(s))} \\
		&= \lim_{\nR \rightarrow \infty} \frac{1}{ \beta^{-(\nR + 1)}(s)(- \log \beta(s))(1-\beta_1(s))} = 0.
		\end{align*}
		In the third step above we have used L'Hospital's Rule. Since the summation is non-negative and is less than or equal to zero, it should be equal to zero. 
	\end{proof}
	It is easy to see that if $s \notin \mathcal{S}$, then $\phi(s,\beta_1(s),\beta_2(s))$ will be equal to infinity. Therefore, using~\eqref{eq6:twohop} and Lemma~\ref{lem:infSum}, we obtain
	\begin{align*}
	&\lim_{t \rightarrow \infty} \P\{D(\nR) > t\} \\
	&\leq \min_{s>0} e^{-s(d \! - \! \frac{1}{R})}M_{1}(s)  M_{2}(s) \phi(s,\beta_1(s),\beta_2(s))\\
	&= \min_{s \in \mathcal{S}} e^{-s(d \! - \! \frac{1}{R})} \cdot \frac{M_{1}(s)  M_{2}(s)}{(1 - \beta_1(s))(1 - \beta_2(s))}.
	\end{align*}
	Hence the result is proven.
	
	\subsection{Proof of Theorem~\ref{thm:twohop:alphaUB}}\label{thm:twohop:alphaUB:proof}
	From Lemma~\ref{lem:twoHop} we have 
	\begin{align}\label{lem4:eq1}
	&\lim_{t \rightarrow \infty} \P\{D(\nR) > t\} \leq \lim_{\nR \rightarrow \infty} \sum_{v_0=0}^{\nR}\sum_{v_1=0}^{\nR - v_0} \Phi(v_0,v_1,R) \nonumber \\
	&= \lim_{\nR \rightarrow \infty} \sum_{v_0=0}^{K - 1}  \sum_{v_{1}=0}^{K-1}\!\! \Phi(v_0,v_1,R) +  \Phi_1(K) + \Phi_2(K),
	\end{align}
	where
	\begin{align*}
	\Phi_1(K) &= \lim_{\nR \rightarrow \infty} \sum_{v_0=0}^{K-1}  \sum_{v_{1}=K}^{\nR - v_0}\Phi(v_0,v_1,R) \\
	\Phi_2(K) &= \lim_{\nR \rightarrow \infty} \sum_{v_0=K}^{\nR}  \sum_{v_{1}=0}^{\nR - v_0}\Phi(v_0,v_1,R).
	\end{align*}
	In the following we use the Chernoff bound for $\Phi(v_0,v_1,R)$, given in~\eqref{eq5:twohop}, to derive bounds for $\Phi_1(K)$ and $\Phi_2(K)$. We have
	\begin{align}\label{lem4:eq2}
	\Phi_1(K) \leq \min_{s>0} e^{-s(d \! - \! \frac{1}{R})}M_{1}(s)  M_{2}(s) \phi_1(s,\beta_1(s),\beta_2(s)),
	\end{align}
	where
	\begin{align}\label{lem4:eq3}
	&\phi_1(s,\beta_1(s),\beta_2(s)) = \lim_{\nR \rightarrow \infty} \! \sum_{v_0=0}^{K-1}\sum_{v_1=0}^{\nR - v_0}\!\! \beta^{v_1}_1(s)\beta^{v_0}_2(s) \nonumber\\
	&= \lim_{\nR \rightarrow \infty} \! \sum_{v_0=0}^{K-1} \beta^{v_0}_2(s) \frac{\beta^K_1(s)(1-\beta^{\nR-v_0+1}_1(s))}{1-\beta_1(s)} \nonumber\\
	&= \frac{\sum_{v_0=0}^{K-1}\!\beta^{v_0}_2(s) \beta_1^K\!(s)}{1-\beta_1\!(s)}\! -\!\! \lim_{\nR \rightarrow \infty} \!\! \sum_{v_0=0}^{K-1} \!\!\frac{\beta^{v_0}_2\!(s)\beta^{\nR-v_0+1+ K}_1\!(s)}{1-\beta_1(s)} \nonumber\\
	&= \frac{(1-\beta^K_2(s))\beta_1^K(s)}{(1-\beta_1(s))(1-\beta_2(s))}, \text{ for } s \in \mathcal{S}.
	\end{align}
	The second term in the third step above vanishes as $\beta_1(s) < 1$ for $s \in \mathcal{S}$. Using~\eqref{lem4:eq3} in~\eqref{lem4:eq2}, we obtain
	\begin{align}\label{lem4:eq4}
	\Phi_1(K)\! \leq \! \min_{s\in \mathcal{S}} e^{-s(d \! - \! \frac{1}{R})}\! M_{1}(s)  M_{2}(s) \frac{(1\!-\!\beta^K_2\!(s))\beta_1^K\!(s)}{(1\!-\!\beta_1\!(s))(1\!-\!\beta_2\!(s))}. 
	\end{align}
	
	Again, substituting~\eqref{eq5:twohop} in $\Phi_2(K)$, we obtain
	\begin{align}\label{lem4:eq5}
	\Phi_2(K) \leq \min_{s>0} e^{-s(d \! - \! \frac{1}{R})}M_{1}(s)  M_{2}(s) \phi_2(s,\beta_1(s),\beta_2(s)),
	\end{align}
	where
	\begin{align*}
	&\phi_2(s,\beta_1(s),\beta_2(s)) = \lim_{\nR \rightarrow \infty} \sum_{v_0=K}^{\nR}  \sum_{v_{1}=0}^{\nR - v_0}\!\! \beta^{v_1}_1(s)\beta^{v_0}_2(s).
	\end{align*}
	Using similar analysis as in Lemma~\ref{lem:infSum}, we obtain
	\begin{align}\label{lem4:eq6}
	\phi_2(s,\beta_1(s),\beta_2(s)) = \frac{\beta^K_2(s)}{(1\!-\!\beta_1\!(s))(1\!-\!\beta_2\!(s))},\!\! \text{ for } s \in \mathcal{S}.
	\end{align}
	Using~\eqref{lem4:eq6} in~\eqref{lem4:eq5}, we obtain
	\begin{align}\label{lem4:eq7}
	\Phi_2(K)\! \leq \!\min_{s\in \mathcal{S}} e^{-s(d \! - \! \frac{1}{R})}\! M_{1}(s)  M_{2}(s) \frac{\beta^K_2(s)}{(1\!-\!\beta_1\!(s))(1\!-\!\beta_2\!(s))}. 
	\end{align}
	Finally, substituting~\eqref{lem4:eq4} and~\eqref{lem4:eq7} in~\eqref{lem4:eq1} we obtain the result.
}
		
\begin{IEEEbiography}[{
		\includegraphics[width=1in,height=1.25in,clip,keepaspectratio]{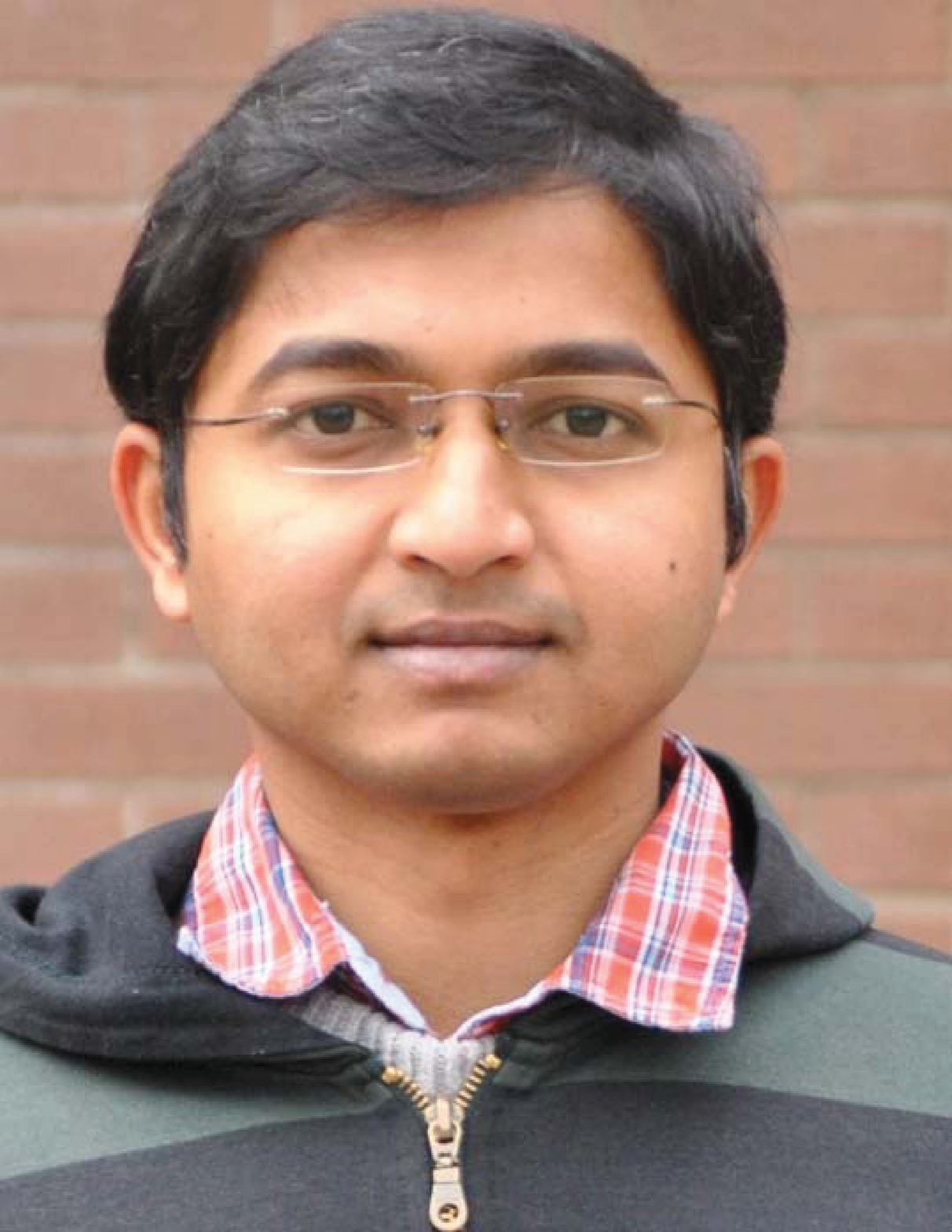}
	}]{Jaya Prakash Champati} is a post-doctoral researcher from the division of Information Science and Engineering, EECS, KTH Royal Institute of Technology, Sweden. He finished his PhD in Electrical and Computer Engineering, University of Toronto, Canada in 2017. He obtained his master of technology degree from Indian Institute of Technology (IIT) Bombay, India, and  bachelor of technology degree from National Institute of Technology Warangal, India. His general research interest is in the design and analysis of algorithms for scheduling problems that arise in networking and information systems. Currently, his focus is on freshness and delay analysis for time-critical control applications in Cyber-Physical Systems (CPS) and the Internet of Things (IoT). In the past, he worked on task scheduling and job assignment problems with motivations from computational offloading in edge computing systems. Prior to joining PhD he worked at Broadcom Communications, where he was involved in developing the LTE MAC layer. He was a recipient of the best paper award at IEEE National Conference on Communications, India, 2011.
\end{IEEEbiography}
\begin{IEEEbiography}
	[{\includegraphics[width=1in,height=1.20in,clip,keepaspectratio]{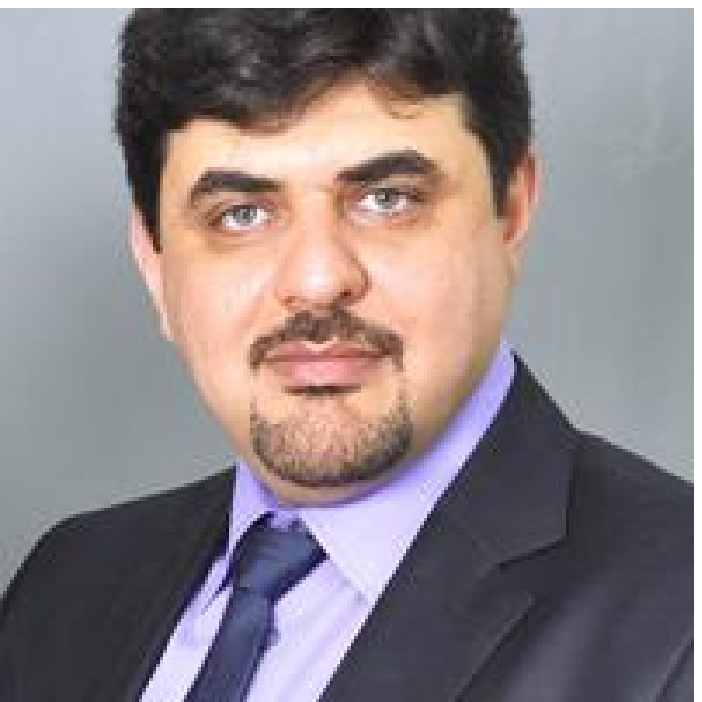}}]
	{Hussein Al-zubaidy} (S07M’11SM’16) received the Ph.D. degree in electrical and computer engineering
	from Carleton University, Ottawa, ON, Canada, in 2010. He was a Post-Doctoral Fellow with the
	Department of Electrical and Computer Engineering, University of Toronto, Toronto, ON, Canada, from
	2011 to 2013. In the Fall of 2013, he joined the School of Electrical Engineering (EES) at the Royal
	Institute of Technology (KTH), Stockholm, Sweden, as a Post-Doctoral Fellow. Since Fall 2015, he has
	been a Senior Researcher with EES at the Royal 	Institute of Technology (KTH), Stockholm, Sweden. Dr. Al-Zubaidy is
	the recipient of many honors and awards, including the Ontario Graduate Scholarship (OGS), NSERC Visiting Fellowship, NSERC Summer Program
	in Taiwan, OGSST, and NSERC Post-Doctoral Fellowship.
\end{IEEEbiography}
\begin{IEEEbiography}
	[{\includegraphics[width=1in,height=1.20in,clip,keepaspectratio]{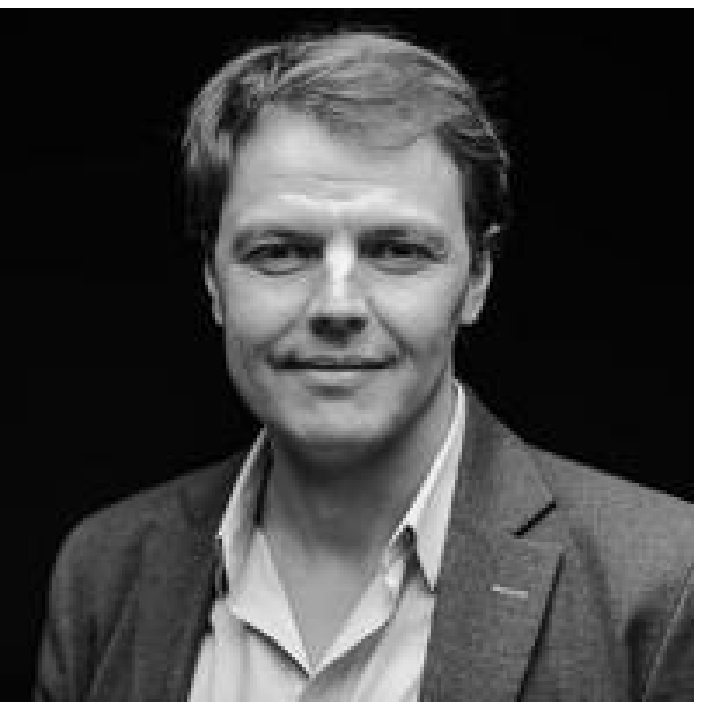}}]
	{James Gross} received his Ph.D. degree from TU Berlin in 2006. From 2008-2012 he was Assistant Professor and head of the Mobile Network Performance Group at RWTH Aachen University, as well as a member of the DFG-funded UMIC Research Centre of RWTH. Since November 2012, he has been with the Electrical Engineering and Computer Science School, KTH Royal Institute of Technology, Stockholm, where he is professor for machine-to-machine communications. He served as Director for the ACCESS Linnaeus Centre from 2016 to 2019, while he is currently a member of the board of KTHs Innovative Centre for Embedded Systems. His research interests are in the area of mobile systems and networks, with a focus on critical machine-to-machine communications, cellular networks, resource allocation, as well as performance evaluation methods. He has authored about 150 (peer-reviewed) papers in international journals and conferences. His work has been awarded multiple times, including the Best Paper Award at ACM MSWiM 2015, the Best Demo Paper Award at IEEE WoWMoM 2015, the Best Paper Award at IEEE WoWMoM 2009, and the Best Paper Award at European Wireless 2009. In 2007, he was the recipient of the ITG/KuVS dissertation award for his Ph.D. thesis. He is also co-founder of R3 Communications GmbH, a Berlin-based start-up in the area of ultrareliable low-latency wireless networking for industrial automation.
\end{IEEEbiography}

\end{document}